\documentclass[journal]{IEEEtran}
\makeatletter
\newcommand*{\rom}[1]{\expandafter\@slowromancap\romannumeral #1@}
\makeatother
\IEEEoverridecommandlockouts
\ifCLASSINFOpdf
\else
\fi

\usepackage{stfloats}
\usepackage{float}
\usepackage{subfloat}
\usepackage{graphics}
\usepackage{multicol}
\usepackage{lipsum}
\usepackage[cmex10]{amsmath}
\usepackage{amsthm}
\usepackage{mathrsfs}
\usepackage{mathtools}
\usepackage{amsbsy}
\usepackage{xcolor}
\usepackage{mathrsfs}
\usepackage{setspace}
\usepackage{graphicx}
\usepackage{lettrine}
\usepackage{newclude}
\usepackage[normalem]{ulem}
\usepackage{latexsym}
\usepackage{algpseudocode}
\usepackage{algorithm,algpseudocode}
\usepackage{algorithmicx}

\usepackage{multirow}
\usepackage{rotating}
\usepackage{booktabs}
\usepackage{wasysym}
\usepackage{mathrsfs}
\usepackage{bbm} 
\usepackage{filecontents}
\usepackage{soul}

\usepackage{optidef}

\newtheorem{lemma}{Lemma}

\newtheorem{remark}{Remark}

\ifCLASSOPTIONcompsoc
\usepackage[caption=false,font=normalsize,labelfon
t=sf,textfont=sf]{subfig}
\else
\usepackage[caption=false,font=footnotesize]{subfig}
\fi
\def\BibTeX{{\rm B\kern-.05em{\sc i\kern-.025em b}\kern-.08em
    T\kern-.1667em\lower.7ex\hbox{E}\kern-.125emX}}
\usepackage{amsmath,epsfig,cite,amsfonts,amssymb,psfrag,color}
\usepackage{epstopdf}

\def\BibTeX{{\rm B\kern-.05em{\sc i\kern-.025em b}\kern-.08em
    T\kern-.1667em\lower.7ex\hbox{E}\kern-.125emX}}

\def\T{\mathrm{T}}
\def\U{\mathrm{U}}

\def\pr{\mathcal{P}}

\def\kp{k^{\prime}}

\def\h{\boldsymbol{h}}
\def\hbar{\bar{\boldsymbol{h}}}
\def\htilde{\tilde{\boldsymbol{h}}}
\def\g{\boldsymbol{g}}
\def\a{\boldsymbol{a}}
\def\v{\boldsymbol{v}}
\def\b{\boldsymbol{b}}

\def\w{\boldsymbol{w}}
\def\u{\boldsymbol{u}}
\def\F{\boldsymbol{F}}
\def\Q{\boldsymbol{Q}}

\newcommand{\abss}[1]{{\left\lvert{#1}\right\rvert}^2}
\newcommand{\abssOne}[1]{{\left\lvert{#1}\right\rvert}}

\newcommand{\norm}[1]{\left\lVert#1\right\rVert}

\begin{document}
\raggedbottom

\title{Resource Allocation in Cooperative Mid-band/THz Networks in the Presence of Mobility}

\author{Mohammad Amin Saeidi, {\em Graduate Student Member IEEE} and  Hina~Tabassum, {\em Senior Member IEEE}
\thanks{The authors are with the Department of Electrical Engineering and Computer Science at York University, Toronto, ON M3J 1P3, Canada. (e-mails: \{amin96a, hinat\}@yorku.ca).}
\thanks{This research was supported by a Discovery Grant funded by the Natural Sciences and Engineering Research Council of Canada.}
}

\maketitle

\begin{abstract}

This paper develops a comprehensive framework to investigate and optimize the downlink performance of cooperative \textit{multi-band networks (MBNs)} operating on upper mid-band (UMB) and terahertz (THz) frequencies, where base stations (BSs) in each band cooperatively serve users.  The framework
captures sophisticated features such as near-field channel modeling, fully and partially connected antenna architectures, and users’ mobility. 
First, we consider joint user association and hybrid beamforming optimization to maximize the system sum-rate, subject to power constraints, maximum cluster size of cooperating BSs, and users' quality-of-service (QoS) constraints.
By leveraging fractional programming FP and majorization-minimization techniques, an iterative algorithm is proposed to solve the non-convex optimization problem.
We then consider handover (HO)-aware resource allocation for \textit{moving} users in a cooperative UMB/THz MBN.
Two HO-aware resource allocation methods are proposed. The first method focuses on maximizing the HO-aware system sum-rate subject to HO-aware QoS constraints. Using Jensen's inequality and properties of logarithmic functions, the non-convex optimization problem is tightly approximated with a convex one and solved. The second method addresses a multi-objective optimization problem to maximize the system sum-rate, while minimizing the total number of HOs. 
Numerical results demonstrate the efficacy of the proposed algorithms, cooperative UMB/THz MBN over stand-alone THz networks, as well as the critical importance of accurate near-field modeling in extremely large antenna arrays. Moreover, the proposed HO-aware resource allocation methods effectively mitigate the impact of HOs, enhancing performance in the considered system.  

\end{abstract}


\section{Introduction}
The upper mid-band (UMB), ranging from 7 to 24 GHz, has recently gained attention as it offers more spectrum than the conventional sub-6 GHz band, while providing better propagation and coverage characteristics compared to millimeter-wave (mmWave) frequencies \cite{Mid-Band-1,Mid-Band-2}. \textcolor{black}{In particular, the UMB operating in the 7.75–8.4 GHz range is intended for both outdoor and indoor applications in urban and suburban areas, aiming to meet high data rate and moderate mobility requirements \cite{Mid-Band-2}.} 
\textcolor{black}{Moreover, frequency bands above 100 GHz, such as the terahertz (THz) band, offer substantial bandwidth to support services that require high data rates \cite{saeidi2024molecularabsorptionaware}. The technical feasibility of frequency bands above 100 GHz for international mobile telecommunications technologies is investigated in the ITU report \cite{itu_report_imt_above_100ghz}.} 
Furthermore, the integration of different frequency bands in multi-band networks (MBNs) can offer improved network sum-rate, enhanced support for mobility, and cost-effective network deployment \cite{MBN-Survey,MBN-Magazine}. MBNs are anticipated to offer superior performance compared to single-frequency networks by opportunistically leveraging the propagation characteristics and available bandwidth of each spectrum. 
\textcolor{black}{While the joint utilization of the UMB and THz frequencies within an MBN has not been studied previously, there are several research works that consider similar MBNs supporting high frequency THz/mmWave and lower frequency bands, such as sub-6 GHz \cite{MBN-Survey,MBN-Magazine,yuan2023joint,MBN-mmWave-6GHz,MBN-ZJ}. Various prototypes and ongoing standardization efforts are also available in \cite{MBN-Survey}.} 
\textcolor{black}{In particular, incorporating lower frequency bands into an MBN alongside THz can be advantageous.} While THz spectrum offers wider transmission bandwidths, the propagation losses caused by the molecular absorption and blockages in line-of-sight (LoS) channels limit the propagation distance \cite{saeidi2024molecularabsorptionaware}. 
Moreover, due to the limited coverage and blockage of the THz band, mobile users frequently switch BSs, resulting in a high number of handovers (HOs) that can degrade system performance. 
\textcolor{black}{MBNs can address the low coverage, higher HOs, and link blockage issues of the THz band through transmission over lower frequencies with better coverage, such as UMB, while simultaneously utilizing the wider THz transmission bandwidth. Leveraging the distinct features of the THz and UMB spectrum, MBNs can enhance overall system performance by improving coverage and reducing  HOs.}

 

Along another note, extremely large antenna arrays (ELAA) are emerging as one of the potential solutions to combat propagation effects at the THz spectrum. An ELAA consists of an array with a large number of antennas and typically utilizes hybrid beamforming instead of traditional digital beamforming, as the latter would require high number of radio frequency (RF) chains.  Additionally, the large aperture of these arrays extends the near-field region due to increased  Fraunhofer distance, necessitating the adoption of near-field communication (NFC) channel modeling \cite{NFC-Tutorial,NFC-Tutorial-2}. In addition to MBN, coordinated multi-point transmissions (CoMP) are another viable solution to ensure reliable connectivity in THz networks where users are supported through multiple BSs on the same resource block. 
\textcolor{black}{In particular, the impact of blockage can be partially mitigated, as users are able to associate with multiple BSs within their cluster. Additionally, cooperative beamforming improves the received power at the users and enables more effective interference cancellation.}
Nonetheless, while CoMP offers improved performance due to spatial diversity in stationary scenarios, its performance gains are unclear for mobile scenarios due to the inherent trade-off between the number of  HOs and cluster size of CoMP\cite{ammar2024handoffs}.

To this end, we develop a novel optimization framework to address the question of ``\textit{which existing wireless technologies can support high frequency THz networks in the presence of users' mobility and to what extent?}" In this regard, we consider investigating three key technologies, i.e.,  UMB/THz MBN, as well as cooperative transmission where BSs in each spectrum cooperate to serve users. The framework captures NFC channels, users' quality-of-service (QoS) requirements, and velocity of users. Thus, our framework addresses questions such as \textit{how cooperative MBN compares to cooperative THz with or without mobility? and how cooperative transmissions enable mobility-aware THz transmissions?}

\subsection{Related Works}

\textcolor{black}{
The utilization of ELAA systems at higher frequencies is essential for addressing the challenge of severe path loss. However, due to the larger array aperture of these antenna systems, it is advantageous to adopt NFC channel modeling based on the spherical wave model \cite{ELAA-NFC-1}. The impact of NFC on the sum-rate maximization problem in multi-user systems is examined in \cite{ELAA-NFC-2}. Compared to conventional far-field models, incorporating NFC introduces an additional degree of freedom by enabling beam focusing in both the angular and distance domains. ELAA and NFC have also demonstrated improved performance in various applications, including modular arrays, simultaneous wireless information and power transfer, and physical layer security \cite{ELAA-NFC-1, ELAA-NFC-2, ELAA-NFC-3, ELAA-NFC-4}. 
}

To date, a few research works considered resource allocation in either single-antenna THz-only networks or RF or sub-6GHz/THz MBNs, where the sub-6~GHz spectrum coexists with the THz spectrum. In \cite{saeidi2024molecularabsorptionaware}, joint power allocation and user association solution was proposed in a dense THz system with multi-connectivity. The solutions are optimized using alternating optimization, combining fractional programming (FP), the alternating direction method of multipliers, and the concept of unimodularity. 
 The works in \cite{MBN-Magazine,MBN-Survey} investigated the impact of RF/THz MBNs on various key performance indicators, including network data rate, deployment cost efficiency, HO, and coverage probabilities. The results indicated promising performance improvements compared to scenarios deploying either only THz or RF BSs.  
In \cite{yuan2023joint}, the authors examined the use of RF/THz MBNs in space-air-ground networks and proposed annealing algorithm to optimize THz and RF channel allocations. In \cite{hassan2020user}, the authors proposed user association strategy to minimize the network load standard deviation while considering users' QoS in RF/THz MBNs.

The aforementioned works focus on \textit{single antenna BSs} without \textit{cooperation among BSs} or \textit{mobility of devices}.

\textcolor{black}{In \cite{fang2021hybrid}, sum power minimization in cooperative mmWave networks with hybrid beamforming is investigated, where analog beamformers are derived using equal-gain transmission, and digital beamformers are optimized through relaxed semidefinite programming.}
Along another note, a variety of research works considered mobility-aware load balancing \cite{MBN-Mobility-2, gupta2024forecaster, zeng2020realistic} or mobility-aware resource allocation \cite{farokhi2018mobility, HoCost-1, prado2023enabling, ren2023handoff} with single-antenna BSs, thereby no  beamforming and no consideration of cooperation among BSs. 
The study in \cite{MBN-Mobility-2} analyzed load balancing in a hybrid Light Fidelity (LiFi)/RF network by examining the impact of mobility on HO efficiency between consecutive states. In \cite{zeng2020realistic}, a hybrid WiFi and LiFi network was explored where dynamic load balancing with HO was considered, and an enhanced evolutionary game theory-based load balancing method was proposed. 
In \cite{farokhi2018mobility}, power, sub-carrier, and cell association were optimized by maximizing the system sum-rate, where the HO success rate is taken into account as a constraint.  The authors of \cite{HoCost-1} proposed HO-aware user association and power allocation to maximize the proportional fairness utility function in a visible light communication  network, assuming knowledge of users' future rates. Furthermore, the authors of \cite{prado2023enabling} investigated proportional fairness maximization by optimizing the ratio of resources a user receives and user association, while incorporating HO cost into the rate formulation. 
Additionally, a delay minimization problem that considers HO interruptions in distributed computing within a high-altitude platform station–assisted vehicular network is addressed in \cite{ren2023handoff}, where task-splitting ratios, transmit power, bandwidth, and computing resource allocation were optimized. \textcolor{black}{The work in \cite{bao2017optimizing} investigates user-centric cooperative networks within a stochastic geometry framework, accounting for both HO cost and data rate. In \cite{ammar2024handoffs}, HO management in user-centric cell-free massive MIMO systems is studied, where the number of HOs is minimized using a non-causal objective function.} 

\subsection{Motivation and Contributions}
To date, none of the aforementioned research works investigated the performance benefits of \textit{coordinated transmissions in MBNs or THz-only networks with or without users' mobility}. Moreover, most of the existing research works rely on single-antenna transmissions, thus the resource allocation solutions do not consider beamforming, and thus not optimized for multi-antenna or ELAA transmissions. This paper develops a comprehensive optimization framework to identify key wireless technologies that can support high frequency THz networks in the presence of users' mobility.  The framework captures sophisticated features such as cooperation in MBNs, users' QoS, hybrid beamforming, and users' mobility. The main contributions can be summarized as follows:

$\bullet$ We propose a cooperative UMB/THz MBN, where BSs are equipped with ELAAs, and formulate the achievable rate based on accurate NFC channel modeling. This model considers spherical wavefronts to capture not only the impact of the distance between users and each antenna element, but also the relative velocity of users with respect to each antenna element. We then formulate  resource allocation problems in two scenarios (a) to maximize the system sum-rate considering \textit{stationary} users and (b)  maximize the HO-aware sum-rate for \textit{moving} users. We optimize analog and digital beamforming, and user association variables, subject to power constraints, maximum cluster size, users' QoS, blockages, etc.
     
$\bullet$ We develop a unified solution to the sum-rate maximization problem for stationary users in cooperative THz/UMB MBN. First, we derive a  closed-form feasible solution for the analog beamformers and reformulate the power budget constraint accordingly. We then transform the  coupled user association and digital beamforming variables  using big-M approach. Using FP  and quadratic transformation, we then handle the non-convexity of the ratio terms in the objective function and constraints. {We prove that applying a quadratic transformation to the FP problem with a logarithmic fractional objective and QoS constraints yields a problem that is equivalent to the original FP.} Additionally, we use the Majorization-Minimization (MM) method to address the binary user association variables. Finally, we transform the non-convex problem into a convex one and propose an iterative algorithm to jointly solve for digital beamformers and user association variables.

    
$\bullet$ {Building on the proposed framework for stationary users, we derive two solutions for the HO-aware resource allocation problem.} In the first solution, we maximize the HO-aware rate subject to HO-aware QoS and soft HO constraints.  {We prove that the HO-aware rate is a \textit{log-concave} function after applying the quadratic transformation. Using this property, we demonstrate that the non-convex optimization problem  can be tightly approximated by a convex problem leveraging Jensen's inequality and logarithmic properties.}
In the second approach, we formulate a multi-objective optimization problem (MOOP) that aims to maximize the system sum-rate, while minimizing the total number of HOs. This problem is transformed into a weighted-objective optimization and solved using the algorithm proposed for the stationary scenario. 
    
$\bullet$ \textcolor{black}{Numerical results validate the effectiveness of the proposed algorithm and quantify the performance gains of  cooperative MBN compared to a stand-alone THz network, both with and without mobility. While requiring more complexity for channel estimation, the results highlight the importance of NFC-based beamforming with ELAAs, as opposed to its FFC-based counterpart. Additionally, it is observed that increasing the user-centric cluster size is not always beneficial in mobility scenarios within cooperative systems.}

\textbf{Notations:} Boldface lowercase letters represent vectors and boldface capital letters to represent matrices. $\mathbb{B}^{B\times K}$ and $\mathbb{R}^{B\times K}$ denote the $B\times K$ space of binary-valued and real-valued numbers, respectively. $\mathbb{C}^{B\times K}$ represent the $B\times K$ space of complex numbers. $\abssOne{\boldsymbol{x}}$ is L2 norm of complex vector $\boldsymbol{x}$ and is the absolute value for scalar values. $\norm{\boldsymbol{X}}_F$ denotes the frobenius norm of matrix $\boldsymbol{X}$. $\boldsymbol{X}^H$ is the conjugate transpose of complex matrix $\boldsymbol{X}$, and $\boldsymbol{X}^{\mathbb{T}}$ is only the transpose of the $\boldsymbol{X}$. $\Re\{\boldsymbol{X}\}$ denotes the real part of $\boldsymbol{X}$, and $x^*$ is the conjugate of complex number $x$. $\boldsymbol{X}=[.,.,.]$ denotes different column of matrix $\boldsymbol{X}$. $\nabla \boldsymbol{f}$ is the gradient of function $\boldsymbol{f}$. $\odot$ and $\otimes$ denote the Hadamard and Kronecker products, respectively. \textcolor{black}{Table~I summarizes the main symbols and notions used throughout the paper.}

\begin{table}[t]
{\color{black}
\caption{List of Main Notations}\label{tab:notation}
\centering
\renewcommand{\arraystretch}{1.0}
\begin{tabular}{ll}
\toprule
Symbol & Definition \\ \midrule
$B^{\T}$, $B^{\U}$ & Number of THz / UMB BSs \\[1pt]
$K$ & Total number of users \\[1pt]
$\mathcal{C}_{k}^{\T}$, $\mathcal{C}_{k}^{\U}$ & User-centric cluster size in THz / UMB  \\[1pt]
$M^{\T}$, $M^{\U}$ & Number of antennas per TBS / UBS \\[1pt]
$f_{\T}$, $f_{\U}$ & Carrier frequency of THz / UMB band \\[1pt]
$\omega^{\T}$, $\omega^{\U}$ & Bandwidth of THz / UMB band \\[1pt]
$\alpha_{b,k}^{n}$ & Association of user $k$ with TBS $b$ at point $n$ \\[1pt]
$\beta_{r,k}^{n}$ & Association of user $k$ with UBS $r$ at point $n$  \\[1pt]
$\xi$ &  Blockage density \\[1pt]
$\mathbf{h}_{b,k}$, $\mathbf{g}_{r,k}$ & THz / UMB channel vectors \\[1pt]
$\mathbf{F}_{b}$, $\mathbf{Q}_{r}$ & Analog beamformers at TBS $b$ / UBS $r$ \\[1pt]
$\mathbf{w}_{b,k}$, $\mathbf{u}_{r,k}$ & Digital beamformers for user $k$ (TBS $b$, UBS $r$) \\[1pt]
$P_{\max}^{\T}$, $P_{\max}^{\U}$ & TBS / UBS transmit-power budgets \\[1pt]
$\gamma_{k}^{\T}$, $\gamma_{k}^{\U}$ & Downlink SINR of user $k$ on THz / UMB \\[1pt]
$R_{k}$ & Aggregate rate of user $k$ (THz+UMB) \\[1pt]
$\bar{R}_{k}$ & HO-aware rate of user $k$ \\[1pt]
$\eta^{\T}$, $\eta^{\U}$ & HO cost factors at THz / UMB \\[1pt]
$\mu_{k}$, $\zeta_{k}$ & FP quadratic-transform auxiliary variables \\ \bottomrule
\end{tabular}
}
\end{table}

\section{System Model and Assumptions}

We consider an MBN, where $B^\T$ THz BSs (TBSs) and $B^\U$ UMB BSs (UBSs) are deployed along a corridor of length $D$ and all BSs are connected to a central processor unit (CPU) through high capacity backhaul links. $\mathcal{B_T}=\{1,\dots,B^\T\}$ and $\mathcal{B_U}=\{1,\dots,B^\U\}$ denote the set of all TBSs and UBSs, respectively. Also, $K$ single-antenna mobile users, are deployed on the corridor such that the velocity of user $k$ is denoted by $v_k$. \textcolor{black}{We assume that each user is equipped with two separate terminals to receive signals at both the UMB and THz spectrum \cite{aboagye2023energy}.}\footnote{\textcolor{black}{In fixed wireless access and vehicular communication applications, as considered in this work, the practical challenges of implementing separate hardware, such as size constraints, are less stringent. Moreover, advanced unified MBN transceivers based on up/down conversion techniques can be utilized to enable the reception of signals across both frequency bands \cite{MBN-Survey,MBN-TRx-1,MBN-TRx-2}. For further discussion of the practical challenges associated with MBNs and potential solution approaches, we refer the reader to \cite{MBN-Magazine,MBN-Survey}.}}

We consider that a cluster of BSs in each frequency band can serve a user in a cooperative manner. Each TBS is equipped with a uniform linear array (ULA) of $M^\T$ antennas connected to a hybrid beamformer (HBF). 
Each HBF at the TBSs comprises of an analog beamformer and a digital beamformer (precoder). The analog beamformer, which is a network of phase shifters, is connected to the digital beamformer via $N^\T$ RF chains ($M^\T \geq N^\T$) either with fully or partially-connected architecture.  Each UBS has $M^\U$ antennas connected to HBF with $N^\U$ RF chains, such that $M^\U \geq N^\U$ \cite{fang2021hybrid}. 

\subsection{Near Field Communication Model}
\begin{figure}
   \centering
\includegraphics[scale=0.71]{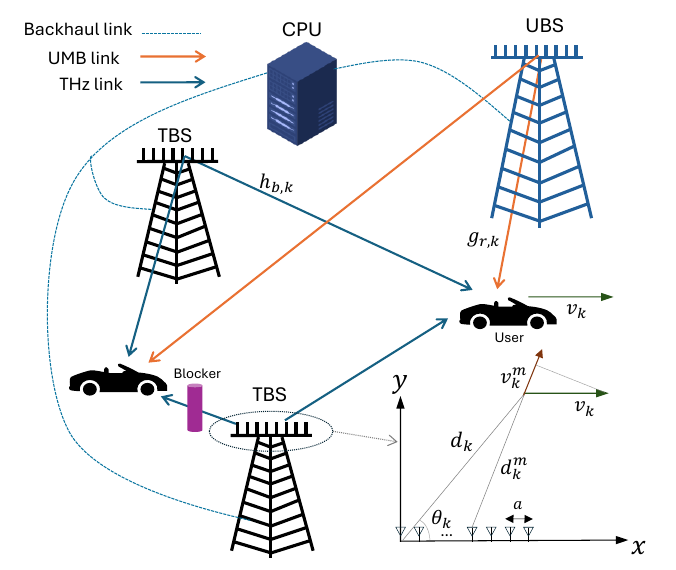}
\caption{{System model illustrating BSs and users and NFC model.}}
\label{fig:SysModel}
\end{figure}

In this subsection, we introduce the NFC modeling regardless of the frequency band, and the notations used here are general. Without loss of generality, according to Fig.~\ref{fig:SysModel}, each ULA is deployed along the $x$-axis, which is in parallel with the users' movement. Denoting the distance between the first antenna element of BS $i\in{\mathcal{B_T} \cup \mathcal{B_U} }$, and user $k$ by $d_{i,k}= d_{i,k}^{(0)}$ and the angle of the user with respect to the first antenna element of a BS by $\theta_{i,k}$, based on the near-field spherical wavefront modeling, the exact distance between user $k$ and $m$-th antenna of BS $i$ is obtained as \cite{NFC-Tutorial}:
\begin{equation}\label{eqn:NFC-Dist}
    d^{(m)}_{i,k}(d_{i,k},\theta_{i,k}) = \sqrt{d_{i,k}^2 + m^2 x^2 - 2d_{i,k} m x \cos(\theta_{i,k})},
\end{equation}
where $0 \leq m \leq M-1$, and $x$ is the antenna spacing, which can be set as half the wavelength of the operating carrier frequency ($\lambda$).
Based on NFC model in \eqref{eqn:NFC-Dist}, the array response of BS $i$ with respect to user $k$ at angle of departure (AoD) of $\theta_{i,k}$ is given by:
\begin{equation}\label{eqn:ArrayResp}
    \a_{i,k}(d_{i,k},\theta_{i,k}) = \left[e^{-j\frac{2\pi}{\lambda}d^{(0)}_{i,k}},\dots,e^{-j\frac{2\pi}{\lambda}d^{(M-1)}_{i,k}} \right]^{\mathbb{T}}.
\end{equation}
Utilizing the distance between user $k$ and the $m$-th antenna of BS $i$, we need to obtain the NFC-aware velocity of user $k$, which can vary with respect to each antenna element. To do so, the projection of $v_k$ onto the line connecting the $m$-th antenna to the user, i.e. $d^{(m)}_{i,k}$ is derived using:
\begin{equation}\label{eqn:Relative-Veloc-NFC}
    v_{i,k}^{(m)}(d_{i,k},\theta_{i,k}) = \frac{d_{i,k} \cos(\theta_{i,k}) - m x}{d^{(m)}_{i,k}} v_k.
\end{equation}
Then, the Doppler-shift of user $k$ at BS $i$ is obtained by:
\begin{equation}\label{eqn:DopplerResp}
    \v_{i,k}(d_{i,k},\theta_{i,k},v_{k}) = \left[e^{-j\frac{2\pi}{\lambda}v_{i,k}^{(0)} \tau},\dots,e^{-j\frac{2\pi}{\lambda}v_{i,k}^{(M-1)} \tau} \right]^{\mathbb{T}},
\end{equation}
where $\tau$ is the duration between two trajectory points.
Based on the array response in \eqref{eqn:ArrayResp} and Doppler-shift in \eqref{eqn:DopplerResp}, we define the velocity-aware effective array response as \cite{THz-Doppler}:
\begin{equation}\label{eqn:Eff-ArrayResp}
    \b_{i,k}(d_{i,k},\theta_{i,k},v_k) = \a_{i,k}(d_{i,k},\theta_{i,k}) \odot \v_{i,k}(d_{i,k},\theta_{i,k},v_{k}).
\end{equation}
From this point onward, the indices $b\in \mathcal{B_T}$ and $r\in \mathcal{B_R}$ are used for the $b$-th TBS and $r$-th UBS, respectively. 

\textcolor{black}{The analysis of parameter estimation for CSI acquisition, such as user location and velocity, and the associated overhead, is beyond the scope of this work. However, data-driven approaches and recent advancements in integrated sensing and communication, as well as cooperative strategies involving low Earth orbit satellites, provide promising directions for improving the estimation of these parameters and enhancing the overall quality of CSI acquisition \cite{Ch-Pred-1,Ch-Pred-2-EKF,Ch-Pred-3-EKF-UserTracking,Ch-Pred-4}. In addition, the overhead associated with CSI transmission can be reduced through dynamic CSI compression techniques, such as the method proposed in \cite{barahimi2024rscnet}.}

\subsection{Channel Model}
\subsubsection{THz Channel}
We consider the LoS transmission links in the THz band. According to \cite{saeidi2024molecularabsorptionaware,MBN-Survey}, and based on Beer-Lambert law, the path-loss between TBS $b$ and user $k$ is given by:
\begin{equation}\label{eqn:THz-PathLoss}
    PL_{b,k} = \left(\frac{c{(G_{\mathrm{Tx}}^{\mathrm{T}} G_{\mathrm{Rx}}^{\mathrm{T}})}^{\frac{1}{2}}}{4\pi f_{\mathrm{T}}d_{b,k}}\right) e^{-\frac{1}{2}\mathcal{K}(f_{\mathrm{T}})d_{b,k}},
\end{equation}
where $c$ is the speed of light, $f_{\mathrm{T}}$ is the carrier frequency of the THz band, and $\mathcal{K}(f_{\mathrm{T}})$ is the molecular absorption coefficient, that can be obtained using high-resolution transmission (HIRTAN) molecular absorption database \cite{THz-Channel-2-HITRAN}. $G_{\mathrm{Tx}}^{\mathrm{T}}$ and $G_{\mathrm{Rx}}^{\mathrm{T}}$ are the gain of each directional antenna element in the arrays of the transmitter and receiver at the THz band, respectively \cite{yan2022energy}. The re-radiation of the absorbed signals by molecules  affects signal propagation, and the signals experience the following path-loss due to molecular absorption noise \cite{saeidi2024molecularabsorptionaware,hossan2021mobility}:
\begin{equation}\label{eqn:THz-PathLoss-MolecNoise}
    \tilde{PL}_{b,k} = \left(\frac{c{(G_{\mathrm{Tx}}^{\mathrm{T}} G_{\mathrm{Rx}}^{\mathrm{T}})}^{\frac{1}{2}}}{4\pi f_{\mathrm{T}}d_{b,k}}\right) {\left(1-e^{-\mathcal{K}(f_{\mathrm{T}})d_{b,k}}\right)}^{\frac{1}{2}}.
\end{equation}

We can then formulate the channels at the THz band by combining the path losses and the effective array response as:
\begin{align}
    & \h_{b,k} = PL_{b,k} \boldsymbol{b}_{b,k}(d_{b,k},\theta_{b,k},v_{k}) \\
    & \htilde_{b,k} = \tilde{PL}_{b,k} \boldsymbol{b}_{b,k}(d_{b,k},\theta_{b,k},v_{k}),
\end{align}
where $\h_{b,k} \in \mathbb{C}^{M^\T\times 1}$ and $\htilde_{b,k}$ are the direct and molecular absorption THz channels between TBS $b$ and user $k$, respectively.

Moreover, THz transmission is susceptible to blockage due to high penetration losses \cite{MBN-Magazine,saeidi2023tractable}; thus, we incorporate link blockage probability. We define $\psi_{b,k} \in {\mathbb{B}}$ as the indicator variable of the THz band for blockage such that $\psi_{b,k} = 0$ if there is a blockage between user $k$ and TBS $b$, and otherwise $\psi_{b,k} = 1$. The blockage variable $\psi$ follows the Bernoulli distribution \cite{Blockage-aware1,AntennaGain}:
\begin{equation}\label{eqn:Blockage-Dist}
\psi_{b,k}= \left\{ \begin{array}{ll}
0, &  \text{With probability} \ \ 1-e^{-\xi d_{b,k}} \\ 1, & \text{With probability} \ \ e^{-\xi d_{b,k}},
\end{array}\right.
\end{equation}
where $\xi$ is the density of blockers in a given area. Considering that if the link between a user and a TBS is blocked, the corresponding channels must have zero vectors. Therefore, the effective THz channels with blockage
are obtained as follows:
\begin{equation} \label{eqn:THz-Ch-Blockage-1}
\h_{b,k} = \psi_{b,k}\mathbf{1}_{M^\T\times 1} \odot\h_{b,k},
\end{equation}
where $\mathbf{1}_{M^\T\times 1}$ is a vector of ones. The same multiplication is also applied to $\tilde{\h}_{b,k}$ as well.
Moreover, in order to impose the blockage on the user association variables, $\alpha_{b,k}^n$, the following constraints need to be considered. 
\begin{equation}
    \alpha_{b,k}^n \leq \psi_{b,k}, \forall b,k.
\end{equation}

\subsubsection{UMB Channel}
The path-loss at the UMB frequency range between user $k$ and UBS $r$ is modeled by:
\begin{equation}\label{eqn:MB-PathLoss}
    {PL}_{r,k} = \frac{c \ {(G_{\mathrm{Rx}}^{\mathrm{\U}}G_{\mathrm{Tx}}^{\mathrm{\U}})}^{\frac{1}{2}}}{4\pi f_{\U}} \left(\frac{1}{{d}_{r,k}}\right)^{\frac{1}{2}\alpha},
\end{equation}
where $G_{\mathrm{Rx}}^{\mathrm{M}}$, $G_{\mathrm{Tx}}^{\mathrm{M}}$ are the antenna gain of the transmitter and receiver at UMB, respectively, $d_{r,k}$ is the distance between UBS $r$ and user $k$, $\alpha$ is the path-loss exponent, and $f_{\mathrm{M}}$ is the UMB carrier frequency. Therefore, the channel between UBS $r$ and user $k$ is given by:
\begin{equation}\label{eqn:MB-Channel}
    \g_{r,k} \hspace*{-1mm}=\hspace*{-0.5mm} {PL}_{r,k}\hspace*{-1mm}\left(\hspace*{-1mm}\sqrt{\frac{\kappa}{1+\kappa}}\boldsymbol{b}_{r,k}(d_{r,k},\theta_{r,k},v_{k}) \hspace*{-0.5mm} +\hspace*{-0.5mm} \sqrt{\frac{1}{1+\kappa}} \g_{r,k}^{\mathrm{NLoS}}\hspace*{-1mm}\right)\hspace*{-0.1mm},
\end{equation}
where $\g_{r,k} \in \mathbb{C}^{M^\U \times 1}$ is the channel between UBS $r$ and $k$-th user, $\kappa$ is the Rician factor, and $\g_{r,k}^{\mathrm{NLoS}}$ is the NLoS component following zero mean and unit variance complex Gaussian distribution.

\subsection{\textcolor{black}{Downlink SINR and Rate Model}}
This subsection explains the signal model in the cooperative MBN system. 
The received signal at the THz terminal of user $k$ is given by:
\begin{align}\label{eqn:THzRxSignal}
    & y_{k}^\T = \underbrace{\sum\limits_{b=1}^{B^\T}\alpha_{b,k}^n\h_{b,k}^H \F_{b} \w_{b,k} s_{k}  \hspace{-1mm}}_{\text{Desired signals}} + \underbrace{ \sum\limits_{\kp\neq k}^{K}\sum\limits_{b=1}^{B^\T} \alpha_{b,\kp}^n \h_{b,k}^H \F_b \w_{b,\kp} s_{\kp} \hspace{-1mm}}_{\text{Multi-user Interference}} \notag \\ & +z^\T +\tilde{z}_k,
\end{align}
where ${s}_{k} \sim \mathcal{CN}(0,1)$ is the independent and identically distributed (i.i.d.) information symbol of user $k$. $\w_{b,k} \in \mathbb{C}^{N^\T\times 1}$ is the digital beamformer, which is responsible for encoding the $k$-th user's message signal. We consider the following HBF architectures, i.e.,
\begin{itemize}
    \item \textbf{Fully-connected HBF:} where all antenna elements are connected to all of the RF chains through a network of $N^\T M^\T$ phase shifters. The analog beamformer at TBS $b$ is denoted by $\F_b = [\boldsymbol{f}_{b,1},\dots,\boldsymbol{f}_{b,N^{\T}}] \in \mathbb{C}^{M^\T\times N^\T}$. 
    \item \textbf{Partially-connected HBF:}  where each RF chain is connected to ${M^\T}/{N^\T}$ sub-array of antennas.  The analog beamformer can then be given by $\F_b = \textrm{blkdiag}\left([\boldsymbol{f}_{b,1},\dots,\boldsymbol{f}_{b,N^{\T}}]\right) \in \mathbb{C}^{M^\T\times N^\T}$ with  $M^\T$ phase shifters, where $\boldsymbol{f}_{b,i} \in \mathbb{C}^{\frac{M^\T}{N^\T} \times 1}$ 
    is the analog beamformer connected to the $i$-th RF chain via $\frac{M^\T}{N^\T}$ phase shifters.   
\end{itemize}

Moreover, in \eqref{eqn:THzRxSignal}, $\alpha_{b,k}^n$ is the association variable, which equals one if user $k$ associates with TBS $b$ at trajectory point $n$ and zero otherwise. $z^\T$ is the additive Gaussian noise with zero mean and variance of $\omega^\T N_0$, such that $\omega^\T$ is the bandwidth at the THz band and $N_0$ is the power spectral density of the additive noise. 
Furthermore, $\tilde{z} \sim \mathcal{CN}(0,\tilde{\sigma}^2)$ denotes the molecular absorption noise, resulting from the re-radiation of transmitted signals by airborne molecules and is modeled as additive Gaussian noise \cite{ye2021modeling}. \textcolor{black}{The molecular absorption noise can be comparable to both thermal noise and interference, and must be accounted for to avoid overestimating the SINR \cite{saeidi2024molecularabsorptionaware}.} Building upon \cite{ye2021modeling,pradhan2023robust}, in the considered cooperative multi-antenna network with coherent transmissions, the variance of the molecular absorption noise at user $k$ is given by: \vspace{-1mm}
\begin{align}\label{eqn:MolecNoise}
    \tilde{\sigma}^2_k = \sum\limits_{i=1}^K \sum\limits_{j=1}^{B^\T} \sum\limits_{m=1}^{B^\T} \htilde_{j,k}^H \boldsymbol{\bar{w}}_{j,i}\boldsymbol{\bar{w}}_{m,i}^H \htilde_{m,k},
\end{align}
where $\boldsymbol{\bar{w}}_{j,i} = \boldsymbol{F}_j \boldsymbol{w}_{j,i}$.  

Considering the interference and molecular absorption noise, the signal-to-interference-plus-noise ratio (SINR) at the THz band for user $k$ is given by:
\begin{equation}\label{eqn:THzSINR}
\begin{split}  
    \gamma_{k}^\T= \frac{\abss{\sum\limits_{b=1}^{B^\T}\alpha_{b,k}^n\h_{b,k}^H \F_{b} \w_{b,k}}}{\sum\limits_{\kp\neq k}^{K} \abss{\sum\limits_{b=1}^{B^\T} \alpha_{b,\kp}^n\h_{b,k}^H \F_b \w_{b,\kp}} + \tilde{\sigma}^2_k + N_0\omega^\T}
\end{split}
\end{equation}

The received signal at UMB terminal of user $k$ is given by: 
\begin{equation}\label{eqn:RFRxSignal}
    {y}_{k}^\U \hspace*{-1mm}=\hspace*{-1mm} \underbrace{\sum\limits_{r=1}^{B^\U}\beta_{r,k}^n\g_{r,k}^H \Q_{r} \u_{r,k} \hat{s}_{k}  \hspace{-1mm}}_{\text{Desired signals}} +\hspace*{-2mm} \underbrace{ \sum\limits_{\kp\neq k}^{K}\sum\limits_{r=1}^{B^\U} \beta_{r,\kp}^n \g_{r,k}^H \Q_r \u_{r,\kp} \hat{s}_{\kp} \hspace{-1mm}}_{\text{Multi-user Interference}} + {z^\U},
\end{equation}
where $\Q_r =[\boldsymbol{q}_{r,1},\dots,\boldsymbol{q}_{r,N^\U}] \in \mathbb{C}^{M^\U\times N^\U}$ denotes the analog beamformer (can be either fully or partially-connected), $\u_{r,k}\in \mathbb{C}^{N^\U\times 1}$ is the digital beamformer for encoding the message of user $k$, i.e., $\hat{s}_{k} \sim \mathcal{CN}(0,1)$, at the $r$-th UBS. $ z^\U$ is the additive thermal noise following a Gaussian distribution with zero mean and variance of $N_0 \omega^\U$, wherein $\omega^\U$ is the channel bandwidth of the UMB. The association variable at point $n$ is denoted by $\beta_{r,k}^n$, such that if user $k$ is associated with UBS $r$, we have $\beta_{r,k}^n=1$, and $\beta_{r,k}^n=0$, otherwise. 


Based on the received signal in \eqref{eqn:RFRxSignal}, the SINR of the UMB at user $k$ is given by: 
\begin{equation}\label{eqn:RF-SINR}
    {\gamma}_{k}^\U = \frac{\abss{\sum\limits_{r=1}^{B^\U}\beta_{r,k}^n\g_{r,k}^H \Q_{r} \u_{r,k}}}{\sum\limits_{\kp\neq k}^{K}\abss{\sum\limits_{r=1}^{B^\U} \beta_{r,\kp}^n \g_{r,k}^H \Q_r \u_{r,\kp}}  + N_0 {\omega^\U}}.
\end{equation}

Based on the SINR terms in \eqref{eqn:THzSINR} and \eqref{eqn:RF-SINR}, and the dual-connectivity from both THz and UMB \cite{aboagye2023energy}, the overall achievable rate at user $k$ is obtained by:
\begin{equation}\label{eqn:RateUser}
    R_k = \omega^\T \log_2(1+\gamma_k^\T) + {\omega^\U} \log_2(1+\gamma^{\U}_k)
\end{equation}

\section{Joint User Association and Hybrid Beamforming in Cooperative THz/UMB MBN}
In this section, we formulate and solve the resource allocation problem with a focus on maximizing the total network sum-rate when users are stationary. 

\subsection{Problem Formulation}
We formulate the downlink\footnote{\textcolor{black}{In uplink, channel reciprocity in time division duplexing can be exploited to perform resource allocation by optimizing the users' power allocation and the receive combiners at the CPU, to maximize the uplink sum-rate.}} sum-rate maximization problem as:
\begin{align}
&\pr_1:\underset{\substack{\alpha^n_{b,k},\w_{b,k},\F_b \\ \beta^n_{r,k},\u_{r,k}, \Q_r}}{\textrm{maximize}} \,\, \,\,U_1 = \sum\limits_{k=1}^K R_k(\alpha_{b,k}^n,\beta_{r,k}^n,\w_{b,k},\u_{b,k},\F_b,\Q_r) \notag \\
&\mbox{s.t.}\hspace*{4mm}
\mathbf{C_1}: \sum\limits_{b=1}^{B^\T} \alpha_{b,k}^n \leq \mathcal{C}^\T_k, \ \forall k, \hspace*{4mm}\mathbf{C_2}: \sum\limits_{r=1}^{B^\U} \beta_{r,k}^n \leq \mathcal{C}^\U_k, \forall k \notag\\
& \mathbf{C_3}:\sum\limits_{k=1}^{K} \alpha_{b,k}^{n}\leq N^\T, \forall b, \hspace*{12mm}
\mathbf{C_4}:\sum\limits_{k=1}^{K} \beta_{r,k}^{n}\leq N^\U, \forall r \notag \\ 
& \mathbf{C_5}: \sum\limits_{k=1}^K\abss{\alpha_{b,k}^n \F_b\w_{b,k}} \leq  P_{\max}^{\T}, \forall b, \notag \\ 
& \mathbf{C_6}: \sum\limits_{k=1}^K\abss{\beta_{r,k}^n \Q_r\u_{r,k}} \leq  P_{\max}^{\U}, \forall r, \hspace*{2mm} \mathbf{C_7}: \alpha_{b,k}^n \leq \psi_{b,k}, \forall b,k, \notag \\  &\mathbf{C_8}: \abssOne{[\F_b]_{i,j}} =1 , \forall i,j,b, \hspace*{2mm} \mathbf{C_{9}}: \abssOne{[\Q_r]_{i,j}} = 1, 
\forall i,j,r, \notag \\  
&\mathbf{C_{10}}: R_k \geq R_{k}^{\mathrm{Th}}, \forall k, \hspace*{2mm} \mathbf{C_{11}}: \alpha_{b,k}^n,\beta_{r,k}^n \in \mathbb{B}, \forall b,r,k.\notag
\end{align}
In the optimization problem $\pr_1$, to reduce the communication overhead, the number of BSs a user can connect to is limited in each frequency band. Hence, $\mathbf{C_{1}}$ and $\mathbf{C_{2}}$ denote the total number of BSs (user-centric cluster size) that a user can associate with in the THz ($\mathcal{C}^\T_k$) and UMB ($\mathcal{C}^\U_k$), respectively. $\mathbf{C_{3}}$ and $\mathbf{C_{4}}$ constrain the number of users served by each BS in the respective frequency band to the number of RF chains at the BS. Constraints $\mathbf{C_{5}}$ and $\mathbf{C_{6}}$ ensure that the transmit power allocated to all users associated with a given BS does not exceed the power budget, $P_{\max}^{\T}$ for THz and $P_{\max}^{\U}$ for UMB, respectively. $\mathbf{C_{7}}$ ensures that a THz link is not assigned to a TBS-user pair if it is blocked. $\mathbf{C_{8}}$ and $\mathbf{C_{9}}$ impose the constant modulus constraints on the analog beamformers at TBSs and UBSs, respectively, ensuring that each phase shifter maintains a constant amplitude. 
Moreover, $\mathbf{C_{10}}$ is the QoS constraint considering a minimum rate of $R_{k}^{\mathrm{Th}}$ for user $k$, and $\mathbf{C_{11}}$ denotes that the association variables are binary. 

The optimization problem in $\pr_1$ is mixed-integer nonlinear programming (MINLP), which is non-convex, NP-hard, and cannot be directly solved in polynomial time using traditional methods. The non-convexity of the problem arises from several factors: (i) The objective function is a non-convex function due to the fractional form and multiplication of the optimization variables, (ii) the power budget constraints in $\mathbf{C_{5}}$ and $\mathbf{C_{6}}$ are non-convex, (iii) the unit modulus constraints in $\mathbf{C_{8}}$ and $\mathbf{C_{9}}$ are non-convex, and (iv) the QoS constraint in $\mathbf{C_{10}}$ is also non-convex. It is worth noting that the QoS constraint can usually be reformulated in a convex format. However, in this problem, due to the different bandwidths in the considered MBN, i.e., $\omega^\T$ and $\omega^\U$, the constraint cannot be reformulated into a convex function. Also, the binary association variables further render the problem intractable. Therefore, our proposed solution methodology is given as: \textbf{(1)} determining feasible analog beamformers to maximize array gain, \textbf{(2)} transforming the coupled user association and digital beamforming variables  using big-M approach, \textbf{(3)} convexifying the SINR terms in the objective and QoS constraints using FP, \textbf{(4)} utilizing MM and penalty methods to address the user association variables, and \textbf{(5)} proposing an iterative algorithm to solve the joint user association and digital beamforming problem.

\subsection{Analog Beamforming}
By formulating the analog beamformer problem as an array gain maximization problem, which is independent of the digital beamformers and user associations, a sub-optimal solution for the analog beamformers $\F_b$ and $\Q_r$ can be determined by extending \cite{fang2021hybrid} for MBNs.
That is, the upper bound of the array gain between TBS $b$ and and user $k$ as well as UBS $r$  and user $k$ can be obtained using Triangle and Cauchy-Schwartz inequality as follows~\cite{fang2021hybrid}:
\begin{align}\label{eq:ArrayGains}
    & \norm{\boldsymbol{B}_{b} \boldsymbol{F}_b \boldsymbol{W}_{b} + \boldsymbol{G}_{r} \boldsymbol{Q}_r \boldsymbol{U}_{r}}_F^2 \leq \norm{\boldsymbol{B}_{b} \boldsymbol{F}_b}_F^2 \norm{\boldsymbol{W}_b}_F^2 \notag \\ & + \norm{\boldsymbol{G}_{r} \boldsymbol{Q}_r}_F^2 \norm{\boldsymbol{U}_r}_F^2 + 2 \Re(\langle \boldsymbol{B}_{b}\boldsymbol{F}_b \boldsymbol{W}_{b} ,\boldsymbol{G}_{r}\boldsymbol{Q}_r \boldsymbol{U}_{r}\rangle) \notag \\ & = \left(\abss{\b_{b,k}^H \boldsymbol{f}_{b,k}} + \norm{(\boldsymbol{B}_b)_{-k,:}(\boldsymbol{F}_b)_{:,-k}}_F^2\right)\norm{\boldsymbol{W}_b}_F^2 \notag \\ & + \left(\abss{\g_{r,k}^H \boldsymbol{q}_{r,k}} + \norm{(\boldsymbol{G}_r)_{-k,:}(\boldsymbol{Q}_r)_{:,-k}}_F^2\right)\norm{\boldsymbol{U}_r}_F^2,
\end{align}
where $\boldsymbol{B}_b = [\b_{b,1},\dots,\b_{b,K}]^H$ (i.e., $\b_{b,k} = \boldsymbol{a}_{b,k}$ for $v_{b,k} =0$), $\boldsymbol{W}_b = [\alpha_{b,1}^n\w_{b,1},\dots,\alpha_{b,K}^n\w_{b,K}]$, $\boldsymbol{G}_r =[ \g_{r,1},\dots,\g_{r,K}]^H$, $\boldsymbol{U}_r = [\beta_{r,1}^n\u_{r,1},\dots,\beta_{r,K}^n\u_{r,K}]$. $(\boldsymbol{X})_{-k,:}$ and $(\boldsymbol{X})_{:,-k}$ denote the sub-matrices of $\boldsymbol{X}$ after removing $k$-th row and $k$-th column, respectively. 
Also, it is noteworthy that the inner product $\langle \boldsymbol{B}_{b}\boldsymbol{F}_b \boldsymbol{W}_{b} ,\boldsymbol{G}_{r}\boldsymbol{Q}_r \boldsymbol{U}_{r}\rangle =0$ due to the orthogonality of THz and UMB frequencies. 

Given $\boldsymbol{W}_b$ and $\boldsymbol{U}_r$, the analog beamforming for each pair of user and BS, which is feasible and maximizes the array gain\footnote{$\norm{(\boldsymbol{B}_b)_{-k,:}(\boldsymbol{F}_b)_{:,-k}}_F^2$ and $\norm{(\boldsymbol{G}_r)_{-k,:}(\boldsymbol{Q}_r)_{:,-k}}_F^2$ in \eqref{eq:ArrayGains} represent the interference which need not be considered while maximizing the array gain.  However, once the feasible analog beamformers are determined, the interference  will be considered in the optimization.
}, is obtained by solving the following problem:    
\begin{align}\label{Analog_SubProb}
&\underset{\boldsymbol{f}_{b,k}, \boldsymbol{q}_{r,k}}{\textrm{maximize}} \,\, \,\,\abss{{{\b}_{b,k}^H} \boldsymbol{f}_{b,k}} + \abss{{{\boldsymbol{g}_{r,k}^H}} \boldsymbol{q}_{r,k}}\\
&\mbox{s.t.}\hspace*{4mm}
\mathbf{C_5},\mathbf{C_6},\mathbf{C_8},\mathbf{C_9} \notag
\end{align}
Solving \eqref{Analog_SubProb} requires the number of RF chains to be equal to the number of users, that is $N^\T = N^\U = K$. Hence, constraints $\mathbf{C_{3}}$ and $\mathbf{C_{4}}$ are inherently  satisfied.
For fully-connected HBF, the solution to \eqref{Analog_SubProb} is given by:

\vspace{1mm}
\begin{minipage}{0.45\columnwidth}
    \begin{equation}
      \bar{\F}_b = \boldsymbol{B}_b^\mathbb{T}
      \label{eqn:THz-AnalogBeam-FC}
    \end{equation}
\end{minipage}
\hfill
\begin{minipage}{0.45\columnwidth}
    \begin{equation}
      \bar{\Q}_r = e^{j\angle \ \boldsymbol{G}_r^\mathbb{T}}
      \label{eqn:MB-AnalogBeam-FC}
    \end{equation}
\end{minipage}

\vspace*{2mm}
\noindent For the case of partially-connected, we have:
\begin{equation}\label{eqn:THz-AnalogBeam-PC}
    \bar{\F}_b = (\boldsymbol{I}_{K} \otimes \boldsymbol{1}_{\frac{M^\T}{K}}) \odot  \boldsymbol{B}_b^\mathbb{T},
\end{equation}
\begin{equation}\label{eqn:MB-AnalogBeam-PC}
\bar{\Q}_r = (\boldsymbol{I}_{K} \otimes \boldsymbol{1}_{\frac{M^\U}{K}}) \odot  e^{j\angle \ \boldsymbol{G}_r^\mathbb{T}}.
\end{equation}
The analog beamformers obtained above also conform to the unit modulus constraints, thus $\mathbf{C_{8}}$ and $\mathbf{C_{9}}$ are satisfied. In order to prove the feasibility  of $\mathbf{C_5}$ and $\mathbf{C_6}$, we first reformulate the constraints as follows: 
\begin{equation} \label{eqn:THz-Power-FrobNorm}
\begin{split}
    \mathbf{C_5}: \norm{\F_b\boldsymbol{W}_b}_F^2 &= \mathrm{tr}(\boldsymbol{W}_b^H \F_b^H \F_b \boldsymbol{W}_b) \\ & = \mathrm{tr}(\F_b^H \F_b \boldsymbol{W}_b \boldsymbol{W}_b^H) \leq P_{\max}^{\T},
\end{split}
\end{equation}
\begin{equation} \label{eqn:RF-Power-FrobNorm}
\begin{split}
    \mathbf{C_6}: \norm{\Q_r\boldsymbol{U}_r}_F^2 &= \mathrm{tr}(\boldsymbol{U}_r^H \Q_r^H \Q_r \boldsymbol{U}_r) \\ & = \mathrm{tr}( \Q_r^H \Q_r \boldsymbol{U}_r \boldsymbol{U}_r^H) \leq P_{\max}^{\U},
\end{split}
\end{equation}
where both \eqref{eqn:THz-Power-FrobNorm} and \eqref{eqn:RF-Power-FrobNorm} utilize cyclic property of trace of matrices. Then, we depict the feasibility of the analog beamformers for $\mathbf{C_{5}}$ and $\mathbf{C_{6}}$ as shown below.
\begin{remark}
    In a hybrid beamformer with $M$ antenna elements and $N$ RF chains, for large values of $M$, the analog beamformer $\F$ satisfies $\F^H\F \approx M \boldsymbol{I}_{N}$ in a fully-connected architecture, and $\F^H\F = \frac{M}{N}\boldsymbol{I}_N$ in a partially-connected architecture \cite{HBF-1}.
\end{remark}
Based on the result of \textbf{Remark 1}, the power budget constraints are reformulated as:
\begin{align}
    \mathbf{C_5}: \norm{\boldsymbol{W}_b}_F^2 \leq \bar{P}_{\max}^{\T}, \hspace*{2mm} \mathbf{C_6}: \norm{\boldsymbol{U}_r}_F^2 \leq \bar{P}_{\max}^{\U},
\end{align}
where $\bar{P}_{\max}^{\T} = \frac{P_{\max}^{\T}}{M^\T}$ and $\bar{P}_{\max}^{\U} = \frac{P_{\max}^{\U}}{M^\U}$ for the case of fully-connected HBF, and $\bar{P}_{\max}^{\T} = \frac{P_{\max}^{\T} K}{M^\T}$ and $\bar{P}_{\max}^{\U} = \frac{P_{\max}^{\U} K}{M^\U}$
for partially-connected HBF. Therefore, the obtained analog beamformers are feasible for $\pr_1$.

\subsection{Problem Transformation}
After obtaining the analog beamformers $\bar{\F}_b$ and $\bar{\Q}_r$, we first address the nonlinearity of the multiplication of the binary variables and beamformers in the objective function and the power budget constraints in $\mathbf{C_{5}}$ and $\mathbf{C_{6}}$, and QoS in $\mathbf{C_{10}}$. It is important to note that multiplying the binary variables by the digital beamformers results in a zero vector if a user is not associated with a given BS, making the constraints non-convex. To tackle this, we utilize the \textit{big}-M method \cite{big-M}, which allows us to decouple the association variables from the SINR terms and the power budget constraints. First, we introduce variables $p_{b,k} = \abss{\w_{b,k}}$ and $\delta_{r,k} = \abss{\u_{r,k}}$, denoting the power of the digital beamformers for each pair of user-BS at TBSs and UBSs, respectively. The following constraints guarantee that the beamformers remain non-zero whenever the association variables are non-zero, and zero, otherwise.
\begin{subequations}
  \begin{minipage}{0.45\columnwidth}
    \begin{equation}
      \abss{\w_{b,k}} \leq \alpha_{b,k}^n p_{b,k}
      \label{eqn:THz-Power-NonConvex}
    \end{equation}
  \end{minipage}
  \hfill
  \begin{minipage}{0.45\columnwidth}
    \begin{equation}
      \abss{\u_{r,k}} \leq \beta_{r,k}^n \delta_{r,k}
      \label{eqn:RF-Power-NonConvex}
    \end{equation}
  \end{minipage}
\end{subequations}
However, both constraints are still non-convex. To tackle this, the right-hand side of \eqref{eqn:THz-Power-NonConvex} and \eqref{eqn:RF-Power-NonConvex} are expanded as follows:
\begin{subequations}
\begin{equation} \label{eqn:THz-Power-NonConvex-2}
\abss{2 \w_{b,k}} \leq {\left(\alpha_{b,k}^n + p_{b,k}\right)}^2 - {\left(\alpha_{b,k}^n - p_{b,k}\right)}^2
\end{equation}
\begin{equation} \label{eqn:RF-Power-NonConvex-2}
\abss{2 \u_{r,k}} \leq {\left(\beta_{r,k}^n + \delta_{r,k}\right)}^2 - {\left(\beta_{r,k}^n - \delta_{r,k}\right)}^2.
\end{equation}
\end{subequations}
After some manipulations, the above constraints can be written in the following convex formulation:
\begin{subequations}
\begin{equation} \label{eqn:THz-Power-Convex}
\abssOne{[2 \w^H_{b,k}, \alpha_{b,k}^n - p_{b,k}]} \leq \alpha_{b,k}^n + p_{b,k},
\end{equation}
\begin{equation} \label{eqn:MB-Power-Convex}
\abssOne{[2 \u^H_{r,k}, \beta_{r,k}^n - \delta_{r,k}]} \leq \beta_{r,k}^n + \delta_{r,k},
\end{equation}
\end{subequations}
Moreover, the power budget constraints in $\mathbf{C_{5}}$ and $\mathbf{C_{6}}$ are reformulated as follows:
\begin{align}
    & \mathbf{C_{5.2}}:\sum\limits_{k=1}^K p_{b,k} \leq \bar{P}_{\max}^{\T}, \forall b, \hspace*{2mm} \mathbf{C_{5.3}}: p_{b,k} \leq \alpha_{b,k}^n  \bar{P}_{\max}^{\T}, \forall b, k \notag \\ & \mathbf{C_{6.2}}: \sum\limits_{k=1}^K \delta_{r,k} \leq \bar{P}_{\max}^{\U}, \forall r, \hspace*{2mm} \mathbf{C_{6.3}}:  \delta_{r,k} \leq \beta_{r,k}^n  \bar{P}_{\max}^{\U}, \forall r, k. \notag
\end{align}
After the above-mentioned transformations and decoupling of the binary variables from the digital beamformers, the functions in $U_1$ and QoS constraints in $\mathbf{C_{10}}$ also reduce to $R_k(\w_{b,k},\u_{b,k},\bar{\F}_b,\bar{\Q}_r)$. For given analog beamformers and after decoupling the association variables and digital beamformers, $\pr_1$ can be reformulated as follows:
\begin{align}
&\pr_2:
\underset{\substack{\alpha^n_{b,k},\w_{b,k},p_{b,k} \\ \beta^n_{r,k},\u_{r,k}, \delta_{r,k}}}{\textrm{maximize}} \,\, \,\,U_1 \\
&\mbox{s.t.}\hspace*{4mm}
\mathbf{C_{5.1}}: \eqref{eqn:THz-Power-Convex}, \mathbf{C_{5.2}},\mathbf{C_{5.3}},\mathbf{C_{6.1}}: \eqref{eqn:MB-Power-Convex}, \mathbf{C_{6.2}},\mathbf{C_{6.3}}, \notag \\ 
& \mathbf{{C}_{1}},\mathbf{{C}_{2}}, \mathbf{{C}_{7}}, \mathbf{{C}_{10}},\mathbf{{C}_{11}}\notag.
\end{align}
Problem $\pr_2$ is still a MINLP due to the non-convexity of the objective function, QoS constraint, and the binary variables. 

\subsection{Joint User Association and  Beamforming}

\subsubsection{FP-based Approach}
First, we develop an equivalent transformation of the general FP problems with fractional objective function and constraints. This transformation leverages on quadratic transformation method introduced in \cite{FP-Part1, FP-Part2}.
The transformation is presented in the following Lemma.
\begin{lemma}\label{Lemma1}
    Given a sequence of logarithm functions of ratios $\frac{A_{i,j}(\boldsymbol{X})}{B_{i,j}(\boldsymbol{X})}$, for $i \in \{1,\dots,I\}$ and $j \in \{1,\dots,J\}$, and $\mathcal{{X}}$ being the convex feasible set of variable $\boldsymbol{X}$, the maximization problem subject to the QoS constraint formulated as:
    \begin{align}\label{Lem1-P1}
        &\underset{\boldsymbol{X}}{{\textrm{maximize}}} \,\, \,\,\sum\limits_{i}\sum\limits_{j} \log\left(\frac{A_{i,j}(\boldsymbol{X})}{B_{i,j}(\boldsymbol{X})}\right)
        \\
        &\mbox{s.t.}\hspace*{4mm}
        \sum\limits_{j} \log\left(\frac{A_{i,j}(\boldsymbol{X})}{B_{i,j}(\boldsymbol{X})}\right) \geq R^{\textrm {Th}}_i, \ \boldsymbol{X} \in \mathcal{X} \notag
    \end{align}

    is equivalent to
    \vspace{-2mm}
    \begin{align}\label{Lem1-P2}
     &\underset{\boldsymbol{X},\boldsymbol{y}}{\textrm{maximize}} \,\, \,\,\hspace*{-3mm}\sum\limits_{i} \sum\limits_{j} \log\hspace*{-0.5mm}\left(\hspace*{-0.5mm}2\Re\left\{y_{i,j}^*\sqrt{A_{i,j}(\boldsymbol{X})}\right\}-\abss{y_{i,j}}B_{i,j}(\boldsymbol{X})\hspace*{-0.5mm}\right)
        \\
        &\mbox{s.t.}
        \sum\limits_{j} \log\hspace*{-0.5mm}\left(\hspace*{-0.5mm}2\Re\left\{y_{i,j}^*\sqrt{A_{i,j}(\boldsymbol{X})}\right\}\hspace*{-0.5mm}-\hspace*{-0.5mm}\abss{y_{i,j}}B_{i,j}(\boldsymbol{X})\hspace*{-0.75mm}\right)\hspace*{-0.5mm} \geq\hspace*{-0.5mm} R^{\textrm{Th}}_i, \forall i \notag,
    \end{align}
    where the optimal $\boldsymbol{y}$ is $y^o_{i,j} = \frac{\sqrt{A_{i,j}(\boldsymbol{X})}}{B_{i,j}(\boldsymbol{X})}$.
\end{lemma}
\begin{proof}
See \textbf{Appendix A}.
\end{proof}
Based on the transformation in \textbf{Lemma \ref{Lemma1}}, for fixed auxiliary variable $\mu_{k} \in \mathbb{C}$, the THz SINR expression at iteration $i$ can be formulated in the following quadratic format:
\begin{equation*}
    \gamma_{k}^{\T,{(i)}} = 2\Re\left\{\mu_{k}^* \sum\limits_{b=1}^{B^\T} \w_{b,k}^H \bar{\F}_{b}^H \h_{b,k} \right\} - \abss{\mu_{k}} \times
\end{equation*}
\begin{equation}\label{eqn:THz-SINR-FP}
    \Bigg( \sum\limits_{\kp\neq k}^{K} \abss{\sum\limits_{b=1}^{B^\T} \h_{b,k}^H \bar{\F}_b \w_{b,\kp}} +\tilde{\sigma}^2_k + N_0\omega^\T\Bigg)
\end{equation}
where the respective auxiliary variable can be obtained as:\vspace{-1mm}
\begin{equation}\label{eqn:THz-QuadAux}
    \mu_{k} = \frac{\sum\limits_{b=1}^{B^\T} \w_{b,k}^H \bar{\F}_{b}^H \h_{b,k}}{\sum\limits_{\kp\neq k}^{K} \abss{\sum\limits_{b=1}^{B^\T} \h_{b,k}^H \bar{\F}_b \w_{b,\kp}} +\tilde{\sigma}^2_k + N_0\omega^\T}.
\end{equation}
By introducing the quadratic auxiliary variable $\zeta_k \in \mathbb{C}$, the quadratic format of the SINR terms in the UMB at iteration $i$ is derived as:\vspace{-4mm}
\begin{equation*}
    \gamma_{k}^{\U,{(i)}}= 2\Re\left\{\zeta_{k}^* \sum\limits_{r=1}^{B^\U} \u_{r,k}^H \bar{\Q}_{r}^H \g_{r,k} \right\} 
\end{equation*}
\begin{equation}\label{eqn:MB-SINR-FP}
    - \abss{\zeta_{k}} \Bigg( \sum\limits_{\kp\neq k}^{K} \abss{\sum\limits_{r=1}^{B^\U} \g_{r,k}^H \bar{\Q}_r \u_{r,\kp}} + N_0\omega^\U\Bigg),
\end{equation}
and the variable $\zeta_k$ is given by:
\begin{equation}\label{eqn:MB-QuadAux}
    \zeta_{k} = \frac{\sum\limits_{r=1}^{B^\U} \u_{r,k}^H \bar{\Q}_{r}^H \g_{r,k}}{\sum\limits_{\kp\neq k}^{K} \abss{\sum\limits_{r=1}^{B^\U} \g_{r,k}^H \bar{\Q}_r \u_{r,\kp}} + N_0\omega^\U}.
\end{equation}
Replacing the non-convex SINR terms in the objective function and the QoS constraints by the derived quadratic format, problem $\pr_2$ reduces to the following convex-binary optimization problem $\pr_3$ at $i$-th iteration.
\vspace{-1mm}
\begin{align}
&\pr_3:
\underset{\substack{\alpha^n_{b,k},\w_{b,k},p_{b,k} \\ \beta^n_{r,k},\u_{r,k}, \delta_{r,k}}}{\textrm{maximize}} \,\,\hspace*{-1.5mm}\sum\limits_{k=1}^K \omega^\T\log(1+\gamma_k^{\T,(i)}) + {\omega^\U}\log(1+\gamma_{k}^{\U,(i)}) \notag\\
&\mbox{s.t.}\hspace*{0.1mm}
\mathbf{C_{10}}: \omega^\T\log(1+\gamma_k^{\T,(i)})\hspace*{-0.5mm} +\hspace*{-0.5mm} {\omega^\U}\log(1+\gamma_{k}^{\U,(i)}) \geq R^{\mathrm{Th}}_k, \forall k, \notag \\
&
\mathbf{{C}_{1}},\mathbf{{C}_{2}}, \mathbf{{C}_{5.1}}, \mathbf{{C}_{5.2}}, \mathbf{{C}_{5.3}}, \mathbf{{C}_{6.1}}, \mathbf{{C}_{6.2}}, \mathbf{{C}_{6.3}}, \mathbf{{C}_{7}},\mathbf{{C}_{11}}\notag.
\end{align}
A globally optimal solution of $\pr_3$ can be obtained by utilizing standard algorithms, such as CVX \cite{cvx} combined with MOSEK solver \cite{mosek}. To deal with the binary variables in $\pr_3$, Mosek employs branch-and-bound method in an iterative manner. Although a globally optimal solution to $\pr_3$ can be achieved, in the worst case, the computational complexity of this method can grow exponentially with a large number of BSs and users in the network.

\subsubsection{{Majorization-minimization}-based Approach}
In order to reduce the computational complexity that arises from the binary variables, $\alpha_{b,k}^n$ and $\beta_{r,k}^n$, we seek an MM-based method \cite{MM-Paper,Hyb-RSMA-NOMA} and penalty approach.
First, by introducing the following constraints replacing $\mathbf{{C}_{11}}$, the association variables can only take binary values:
\begin{equation*}
    \mathbf{C_{11.1}}: 0\leq \alpha_{b,k}^n \leq 1 , \forall b, k, \hspace*{0mm}
\mathbf{C_{11.2}}\hspace*{-0.5mm}:\hspace*{-1.5mm} \sum\limits_{b=1}^{B^\T}\sum\limits_{k=1}^K \hspace*{-1mm}\left(\alpha_{b,k}^n - {\alpha_{b,k}^n}^2\right)\hspace*{-0.5mm} \leq\hspace*{-0.5mm} 0,
\end{equation*}
\begin{equation}
    \mathbf{C_{11.3}}: 0\leq \beta_{r,k}^n \leq 1 , \forall r, k, \hspace*{0mm}
\mathbf{C_{11.4}}\hspace*{-0.5mm}: \hspace*{-1.5mm}\sum\limits_{r=1}^{B^\U}\sum\limits_{k=1}^K \hspace*{-1mm}\left(\beta_{r,k}^n - {\beta_{r,k}^n}^2\right)\hspace*{-0.5mm} \leq \hspace*{-0.5mm}0,\notag
\end{equation}
where $\mathbf{{C}_{11.2}}$ and $\mathbf{{C}_{11.4}}$ are non-convex. 
To overcome this, we use the penalty method by penalizing the objective function by adding these two sets of constraints with penalty factors of $\Gamma_1 \gg 0$ and $\Gamma_2 \gg 0 $ for $\mathbf{{C}_{11.2}}$ and $\mathbf{{C}_{11.4}}$, respectively. Therefore, problem $\pr_3$ can be reformulated into $\pr_4$ as follows:
\begin{align}
 &\pr_4:
\hspace*{2mm}\underset{\substack{\alpha^n_{b,k},\w_{b,k},p_{b,k} \\ \beta^n_{r,k},\u_{r,k}, \delta_{r,k}}}{\textrm{maximize}} \,\, \,\,U_1 - \Gamma_1 \mathcal{Z}_1 -\Gamma_2\mathcal{Z}_2\notag \\
&\mbox{s.t.}\hspace*{4mm}
\mathbf{{C}_{1}},\mathbf{{C}_{2}}, \mathbf{{C}_{5.1}}, \mathbf{{C}_{5.2}}, \mathbf{{C}_{5.3}}, \mathbf{{C}_{6.1}}, \mathbf{{C}_{6.2}}, \mathbf{{C}_{6.3}}, \notag \\
&\mathbf{{C}_{7}},\mathbf{{C}_{10}}, \mathbf{{C}_{11.1}},\mathbf{{C}_{11.3}}\notag,
\nonumber
\end{align}
where $\mathcal{Z}_1 = \sum\limits_{b=1}^{B_T}\sum\limits_{k=1}^K \alpha_{b,k}^n - {\alpha_{b,k}^n}^2$ and $\mathcal{Z}_2 = \sum\limits_{r=1}^{B_R}\sum\limits_{k=1}^K \beta_{r,k}^n - {\beta_{r,k}^n}^2$.
The objective function in $\pr_4$ is still non-convex due to adding the penalty functions. At this point, the MM approach is utilized, where the second quadratic term of the penalty function is replaced by its first-order Taylor approximation.
\begin{equation}\label{eqn:MM-method}
    \bar{\boldsymbol{f}}^{(i)}(x) \triangleq \boldsymbol{f}(x^{(i-1)}) + \nabla^{\mathbb{T}} \boldsymbol{f}\left(x^{(i-1)}\right)\left(x-x^{(i-1)}\right),
\end{equation}
where $x = \alpha_{b,k}^n$ for the THz association variables and $x = \beta_{r,k}^n$ for the UMB association variables. Thus, $\pr_5$ can be approximated by the following problem at iteration $i$.
\begin{align}
 &\pr_5:
\underset{\substack{\alpha^n_{b,k},\w_{b,k},p_{b,k} \\ \beta^n_{r,k},\u_{r,k}, \delta_{r,k}}}{\textrm{maximize}} \,\, \,\,U_1^{(i)} - \Gamma_1 \mathcal{L}_1^{(i)} - \Gamma_2 \mathcal{L}_2^{(i)}\\
&\mbox{s.t.}\hspace*{4mm}
\mathbf{{C}_{1}},\mathbf{{C}_{2}}, \mathbf{{C}_{5.1}}, \mathbf{{C}_{5.2}}, \mathbf{{C}_{5.3}}, \mathbf{{C}_{6.1}}, \mathbf{{C}_{6.2}}, \mathbf{{C}_{6.3}}, \notag \\
&\mathbf{{C}_{7}},\mathbf{{C}_{10}}, \mathbf{{C}_{11.1}},\mathbf{{C}_{11.3}}\notag.
\end{align}
where \vspace{-3mm}
\begin{subequations}
\begin{equation} \label{eqn:PenaltyAssoc-THz}
\mathcal{L}_1^{(i)}=\sum\limits_{b=1}^{B^\T}\sum\limits_{k=1}^K\left(\alpha_{b,k}^n\left(1-2\alpha_{b,k}^{n,(i-1)}\right)+{\left(\alpha_{b,k}^{n,(i-1)}\right)}^2\right),\notag
\end{equation}
\begin{equation} \label{eqn:PenaltyAssoc-MB}
\mathcal{L}_2^{(i)}=\sum\limits_{r=1}^{B^\U}\sum\limits_{k=1}^K\left(\beta_{r,k}^n\left(1-2\beta_{r,k}^{n,(i-1)}\right)+{\left(\beta_{r,k}^{n,(i-1)}\right)}^2\right),\notag
\end{equation}
\end{subequations}
and superscript $(i-1)$ denotes the value of the respective variable at the previous iteration. Problem $\pr_5$ is a convex optimization that can be solved using CVX \cite{cvx}. The proposed solution for solving $\pr_1$ is summarized in \textbf{Algorithm~1}.
\begin{algorithm}[]
    \caption{Joint User Association and Hybrid Beamforming Solution to $\pr_1$ and $\hat \pr_1$}
    \label{Algorithm-1}
    \textbf{Input}: Index of trajectory point, $n$, maximum number of iterations $L_{\max}$, stopping accuracy $\epsilon_1$ and $\epsilon_2$, penalty factors, $\Gamma_1$ and $\Gamma_2$. Initialize $\w_{b,k},\u_{r,k},\alpha_{b,k}^n,\beta_{r,k}^n, \gamma_{k}^{\T,(0)},\gamma_{k}^{\U,(0)} \ \forall b,r,k$ with feasible values.
    \begin{algorithmic}[1]  
    \State Obtain the analog beamformers $\bar{\F}_b, \ \forall b$ and $\bar{\Q}_r, \ \forall r$
    \If {Fully-connected HBF}
        \State Use \eqref{eqn:THz-AnalogBeam-FC} and \eqref{eqn:MB-AnalogBeam-FC} 
        \Else 
        \State Use \eqref{eqn:THz-AnalogBeam-PC} and \eqref{eqn:MB-AnalogBeam-PC} 
        \EndIf

        \For {$i=0,1,...,L_{\max}$}
        \State Update $\mu_{k}^{(i)} \forall k$ using \eqref{eqn:THz-QuadAux}
        \State Update $\zeta_{k}^{(i)} \forall k$ using \eqref{eqn:MB-QuadAux}
        \If {Fully-connected HBF}
        \State Use $\bar{P}_{\max}^{\T} = \frac{P_{\max}^{\T}}{M^\T}$ and $\bar{P}_{\max}^{\U} = \frac{P_{\max}^{\U}}{M^\U}$
        \Else 
        \State Use $\bar{P}_{\max}^{\T} = \frac{P_{\max}^{\T} K}{M^\T}$ and $\bar{P}_{\max}^{\U} = \frac{P_{\max}^{\U} K}{M^\U}$
        \EndIf
        \State Solve problem $\pr_5$ and obtain $\w_{b,k}^{(i)}$, $\u_{r,k}^{(i)}$, $\alpha_{b,k}^{n,(i)}$, $\beta_{r,k}^{n,(i)}, \ \forall b,r,k$
        \State Update $\gamma_{k}^{\T,(i)}$ using \eqref{eqn:THz-SINR-FP}
        \State Update $\gamma_{k}^{\U,(i)}$ using \eqref{eqn:MB-SINR-FP}
        \State \textbf{Until} $|U_1^{(i)} - U_1^{(i-1)}| < \epsilon_1$ \& $|(\mathcal{L}_1^{(i)}+\mathcal{L}_2^{(i)}) - (\mathcal{L}_1^{(i-1)}+\mathcal{L}_2^{(i-1)})| < \epsilon_2$
        \EndFor
    \end{algorithmic}
    \textbf{Output}: The optimal resource allocation at point $n$: $\w_{b,k}^{o},\F_b^{o},\u_{r,k}^{o},\Q_r^{o},\alpha_{b,k}^{n,o},\beta_{r,k}^{n,o}, \ \forall b,r,k$.
\end{algorithm}

\subsection{\textcolor{black}{Convergence and Complexity Analysis of Algorithm 1}}
\subsubsection{\textcolor{black}{Convergence}}
\textcolor{black}{
To show the convergence of the iterative part of \textbf{Algorithm 1}, from Step 7 to Step 19, we first denote the variables of $\pr_4$ by $\boldsymbol{x} = {\boldsymbol{w}_{b,k}, \boldsymbol{u}_{r,k}}, \forall b, r, k$, and $\boldsymbol{\alpha} = {\alpha_{b,k}^n, \beta_{r,k}^n}, \forall b, r, k$. Let $\mathcal{X}$ denotes the feasible set, and define the objective function as $f(\boldsymbol{x}, \boldsymbol{\alpha}) = U_1(\boldsymbol{x}) + P(\boldsymbol{\alpha})$, where $U_1$ is the sum-rate before applying the quadratic transformation, and $P(\boldsymbol{\alpha}) = -\Gamma_1 \mathcal{Z}_1 - \Gamma_2 \mathcal{Z}_2 \leq 0$, for all $\boldsymbol{\alpha} \in [0,1]$. At the $i$-th iteration, the sum-rate is lower-bounded by $U_1^{(i)}(\boldsymbol{x}) \leq U_1(\boldsymbol{x})$ via the quadratic transformation (see \eqref{Lem1-Eq1}), such that $U_1^{(i)}(\boldsymbol{x}^{(i-1)}) = U_1(\boldsymbol{x}^{(i-1)})$. The penalty term, being a convex function, is lower-bounded by its first-order Taylor expansion as $P^{(i)}(\boldsymbol{\alpha}) \leq P(\boldsymbol{\alpha}) \leq 0$, where $P^{(i)}(\boldsymbol{\alpha}^{(i-1)}) = P(\boldsymbol{\alpha}^{(i-1)})$. Accordingly, the surrogate objective function at iteration $i$ is defined as: $f^{(i)}(\boldsymbol{x},\boldsymbol{\alpha}) = U_1^{(i)}(\boldsymbol{x}) + P^{(i)}(\boldsymbol{\alpha})$, which forms the objective of $\pr_5$.
}

\textcolor{black}{
The solution at the $i$-th iteration is given by: $\boldsymbol{x}^{(i)},\boldsymbol{\alpha}^{(i)}= \arg\max_{\boldsymbol{x},\boldsymbol{\alpha} \in \mathcal{X}} f^{(i)}(\boldsymbol{x},\boldsymbol{\alpha})$. Since $\pr_5$ is a convex optimization problem, it satisfies the inequality:}
{\color{black}
\begin{align}
    f(\boldsymbol{x}^{(i)},\boldsymbol{\alpha}^{(i)}) \geq f^{(i)}(\boldsymbol{x}^{(i)},\boldsymbol{\alpha}^{(i)}) & \geq f^{(i)}(\boldsymbol{x}^{(i-1)},\boldsymbol{\alpha}^{(i-1)}) \notag \\ & = f(\boldsymbol{x}^{(i-1)},\boldsymbol{\alpha}^{(i-1)}),
\end{align}
} \textcolor{black}{
which shows that the objective of $\pr_4$ increases monotonically with iterations.
Furthermore, since all variables are bounded due to power and QoS constraints, the sequence ${f(\boldsymbol{x}^{(i)}, \boldsymbol{\alpha}^{(i)})}$ is finite. Also, as $P(\boldsymbol{\alpha}^{(i)}) \leq 0$, the objective function of $\pr_4$ is bounded above by the sum-rate: $f(\boldsymbol{x}^{(i)},\boldsymbol{\alpha}^{(i)}) \leq U_1(\boldsymbol{x}^{(i)})$. If sufficiently large penalty factors $\Gamma_1$ and $\Gamma_2$ are chosen, then $P(\boldsymbol{\alpha}^{(i)}) \to 0$, and equality holds. Consequently, $U_1(\boldsymbol{x}^{(i)})$ converges to a sub-optimal point of $\pr_2$.}

\subsubsection{\textcolor{black}{Complexity Analysis}}
\textcolor{black}{
The worst-case complexity of \textbf{Algorithm 1} is obtained as follows. In the first stage of the algorithm, the complexity is dominated by the analog beamforming derivations. If fully-connected HBF is selected, the complexity is $\mathcal{O}((B^\T M^\T + B^\U M^\U)K)$, and if partially connected HBF is chosen, the complexity reduces to $\mathcal{O}(B^\T M^\T + B^\U M^\U)$.
For the iterative part of the algorithm, since only one convex optimization problem is solved in each iteration, the algorithm exhibits polynomial computational complexity \cite{razaviyayn2013unified}. Steps 8, 9, 16, and 17 have a combined complexity of $\mathcal{O}(4K)$, which is dominated by the complexity of Step 15. 
Complexity of solving the convex optimization in step 15 is $\mathcal{O}(((K^2 + 2K)(B^\T + B^\U))^3)$, which is of cubic order in terms of the number of variables \cite{ye2011interior}.}\footnote{\textcolor{black}{The MOSEK implementation is based on the homogeneous model, the Nesterov–Todd search direction, and a Mehrotra-like predictor-corrector algorithm \cite{andersen2003implementing}. These features enhance numerical robustness and practical efficiency. Although interior-point methods have cubic theoretical complexity, MOSEK leverages sparsity and structure to reduce this complexity. Based on MOSEK's FLOP logs for our problem, the total computational complexity scales approximately as $\mathcal{O}(((K^2 + 2K)(B^\T + B^\U))^2)$, which is the square of the number of variables.}} \textcolor{black}{ Therefore, the overall complexity of \textbf{Algorithm 1} is $\mathcal{O}((B^\T M^\T + B^\U M^\U)K + I_A(((K^2+2K)(B^\T+B^\U))^3)$ for fully-connected HBF, and $\mathcal{O}(B^\T M^\T + B^\U M^\U + I_A(((K^2+2K)(B^\T+B^\U))^3)$ for partially-connected HBF, where $I_A$ is the number of iterations required for Steps 7–19 to converge. 
}

\normalsize
\section{Mobility-aware User Association and Hybrid Beamforming in Cooperative THz/UMB MBN}

In this section, we address mobility-aware resource allocation in a cooperative MBN. In traditional single-band multi-BS networks,  HO occurs whenever a user switches from one BS to another. Each HO reduces effective data rate by decreasing the time available for data transmission. In cooperative networks, the number of HOs can be higher since users have more BSs to associate with simultaneously. However, the data rate reduction may not necessarily happen due to cooperation gain of the BSs' cluster. In this section, we develop two approaches to handle HO-aware resource allocation in order to mitigate the impact of HOs on cooperative MBN. 

First, we define the HO-aware effective sum-rate, which serves as the key performance metric that accounts for mobility and HO cost. 
\textcolor{black}{
We assume that a HO occurs at trajectory point $n$ whenever a user associates with a BS at point $n$ that it was not associated with at the previous trajectory point $n-1$. For the $k$-th user, the total number of HOs at point $n$ is given by $\varsigma_k^\T = \sum\limits_{b=1}^{B^\T} \left(1 - \alpha_{b,k}^{n-1}\right) \alpha_{b,k}^{n}$ for the THz band, and by $\varsigma_k^\U = \sum\limits_{r=1}^{B^\U} \left(1 - \beta_{r,k}^{n-1}\right) \beta_{r,k}^{n}$ for the UMB band.
Each HO results in a transmission interruption by reducing the available time for data transmission, which in turn decreases the effective rate at point $n$ \cite{gupta2024forecaster,sun2020movement,HoCost-1,MBN-Mobility-2}.
This interruption is modeled by cost factors $\eta^\T \in (0,1)$ and $\eta^\U \in (0,1)$ for the THz and UMB bands, respectively. These cost factors represent the fraction of transmission time lost due to HO delays, and their values may vary depending on the specific HO procedures implemented by the operator  \cite{gupta2024forecaster}.
Therefore, the remaining time available for transmission over THz and UMB frequency bands is $1 - \eta^\T \varsigma_k^\T$ and $1 - \eta^\U \varsigma_k^\U$, respectively. Based on this, the effective rate of user $k$ at point $n$ is given by $\bar{R}_k = \bar{R}_k^\T + \bar{R}_k^\U$, where
}
\begin{equation}\label{eqn:THz-HO-Rate}
    \hspace{-2mm}\bar{R}_k^{\T} \!=\! \Bigg[\!\!\left(\!1-\eta^\T \sum\limits_{b=1}^{B^\T}\!\left(\!1-\alpha_{b,k}^{n-1}\!\right)\!\alpha_{b,k}^{n}\!\right)\!
    \omega^\T\!\log_2(1+\gamma_k^\T)\Bigg]^+
\end{equation}
and \vspace{-2mm}
\begin{equation}\label{eqn:RF-HO-Rate}
    \hspace{-2mm}\bar{R}_k^{\U} \!=\! \Bigg[\!\!\left(\!1- \eta^\U\sum\limits_{r=1}^{B^\U}\!\left(\!1-\beta_{r,k}^{n-1}\!\right)\!\beta_{r,k}^{n}\!\right)\!\omega^\U\! \log_2(1+\gamma^\U_k)\Bigg]^+,
\end{equation}
where $[x]^+ = \max\{0,x\}$.
The HO-aware system sum-rate is then formulated as follows:
\vspace{-0.1mm}
\begin{equation}\label{eqn:HO-aware-SumRate}
    U_2 = \sum\limits_{k=1}^K \bar{R}_k = \sum\limits_{k=1}^K \bar{R}_k^{\T} + \bar{R}_k^{\U}.
\end{equation}

\subsection{HO Cost-based Approach}
In this approach, we incorporate the direct impact of HO on the system sum-rate by setting the HO-aware sum-rate as the objective function. Hence, the HO cost-based resource allocation problem is formulated as follows: 
\begin{align}
\label{prob-hat1} &\hat{\pr}_1:
\underset{\substack{\alpha^n_{b,k},\w_{b,k},\F_b \\ \beta^n_{r,k},\u_{r,k},\Q_r}}{\textrm{maximize}} \,\, \,\,U_2  \\
&\mbox{s.t.}\hspace*{4mm}
\mathbf{{C}_{1}},\mathbf{{C}_{2}}, \mathbf{{C}_{3}}, \mathbf{{C}_{4}}, \mathbf{{C}_{5}}, \mathbf{{C}_{6}}, \mathbf{{C}_{7}}, \mathbf{{C}_{8}}, \mathbf{{C}_{9}}, \mathbf{{\bar{C}}_{10}}: \bar{R}_k \geq R_k^{\mathrm{Th}},\notag \\
&\mathbf{{C}_{11}},\mathbf{{C}_{12}}: \sum\limits_{b=1}^{B^{\T}} \alpha_{b,k}^{n-1}\alpha_{b,k}^{n} + \sum\limits_{r=1}^{B^{\U}} \beta_{r,k}^{n-1}\beta_{r,k}^{n} \geq S^{\mathrm{HO}},\notag
\end{align}
where $\mathbf{{\bar{C}}_{10}}$ denotes the HO-aware QoS constraint. In addition, we added constraint $\mathbf{C_{12}}$ to ensure that each user maintains a connection to at least $S^{\mathrm{HO}}$ BSs from the previous point.
On top of the challenges identified in Section III.A for solving $\pr_1$, such as non-convexity and the presence of integer variables, we encounter an additional obstacle in solving $\hat{\pr}_1$. This arises from the multiplication of binary variables by achievable rates, specifically in the terms associated with HO costs. It is worth noting that although the binary variables can be decoupled from the beamformers as demonstrated in Section III.C, this approach does not allow for the decoupling of binary variables in the HO cost terms in $U_2$. \textcolor{black}{Therefore, \textbf{Algorithm 1} cannot be directly employed to solve $\hat{\pr}_1$.}

In order to solve $\hat{\pr}_1$, after obtaining the analog beamformer as in Section III.B and transforming the problem as in Section III.C, we apply quadratic programming and MM approach  at $i$-th iteration of the algorithm resulting in the following: 
\begin{align}
\label{prob-hat2} &\hat{\pr}_2:
\underset{\substack{\alpha^n_{b,k},\w_{b,k},p_{b,k} \\ \beta^n_{r,k},\u_{r,k}, \delta_{r,k}}}{\textrm{maximize}} \,\, \,\,U_2^{(i)} - \Gamma_1 \mathcal{L}_1^{(i)} - \Gamma_2 \mathcal{L}_2^{(i)}\\
&\mbox{s.t.}\hspace*{4mm}
\mathbf{{C}_{1}},\mathbf{{C}_{2}}, \mathbf{{C}_{5.1}}, \mathbf{{C}_{5.2}}, \mathbf{{C}_{5.3}}, \mathbf{{C}_{6.1}}, \mathbf{{C}_{6.2}}, \mathbf{{C}_{6.3}}, \notag \\
&\mathbf{{C}_{7}},\mathbf{{\bar{C}}_{10}}, \mathbf{{C}_{11.1}},\mathbf{{C}_{11.3}},\mathbf{{C}_{12}}\notag,
\end{align}
where $U_2^{(i)}$ is the HO-aware sum-rate utility at $i$-th iteration of the Algorithm. \textcolor{black}{$\hat{\pr}2$ remains a non-convex problem due to the multiplication of variables in both the objective function and constraint $\mathbf{{\bar{C}}{10}}$.} 

\begin{lemma}\label{lemma-2}
After applying the quadratic transformations in \eqref{eqn:THz-SINR-FP} and \eqref{eqn:MB-SINR-FP}, each HO-aware rate expression, i.e., $\bar{R}_k^{\T}$ for $\mathbf{C_{13}}:0\leq \sum\limits_{b=1}^{B^\T}\left(1-\alpha_{b,k}^{n-1}\right)\alpha_{b,k}^{n} < \frac{1}{\eta^\T}$ and $\bar{R}_k^{\U}$ for $ \mathbf{C_{14}}:0\leq\sum\limits_{r=1}^{B^\U}\left(1-\beta_{r,k}^{n-1}\right)\beta_{r,k}^{n}<\frac{1}{\eta^\U}$, is a log-concave function.
\end{lemma}

\begin{proof}
    See \textbf{Appendix B}.
\end{proof}

In \textbf{Lemma}~\ref{lemma-2}, we showed that each HO-aware rate is a log-concave function. However, since the sum of log-concave functions is not necessarily log-concave \cite{boyd2004convex}, the objective function $U_2^{(i)}$ and QoS constraints (sum over THz and UMB rates per user) are still non-convex. First, in order to handle the objective function, we propose the following lemma:
\begin{lemma}\label{lemma-3}
    The objective $U_2^{(i)}$ can be replaced with the following concave lower bound:
        $\tilde{U}_2^{(i)} = \frac{1}{2K} \sum\limits_{k=1}^K \log(\bar{R}^\T_k)+\log(\bar{R}^\U_k)+1$.
\end{lemma}
\vspace{-3mm}
\begin{proof} The lower bound of $U_2^{(i)}$ can be obtained as follows: 
\vspace{-2mm}
    \begin{align}
    \tilde{U}_2^{(i)} &= \frac{1}{2K} \!\sum\limits_{k=1}^K \log(\bar{R}^\T_k)\!+\!\log(\bar{R}^\U_k)+\!1\! \overset{(a)}{\leq} \log(\sum\limits_{k=1}^K \frac{1}{2K} \bar{R}_k)\! +\! 1 \notag \\ & \overset{(b)}{\leq} \frac{1}{2K} \sum\limits_{k=1}^K \bar{R}_k \overset{(c)}{\leq} \sum\limits_{k=1}^K \bar{R}_k = U_2^{(i)}.
\end{align}
Following Jensen's inequality for concave function $f(x)$, we have:
\begin{equation}\label{eqn:Jensen-Ineq}
    \lambda_1 f(x_1)+\dots+\lambda_K f(x_K) \leq f(\lambda_1 x_1 +\dots+\lambda_K x_K).
\end{equation}
By setting $\lambda_k = \frac{1}{2K} \ \forall k$ and having $x_k = \bar{R}_k$, inequality $(a)$ can be obtained. Moreover, the inequality in $(b)$ is justified using the property of the logarithm function: $x-1 \geq \log(x), \ \forall x>0$, and $(c)$ is trivial. $\tilde{U}_2^{(i)}$ is also a concave function since $\bar{R}^\T_k$ and $\bar{R}^\U_k,  \forall \ k$ are log-concave functions as proved in \textbf{Lemma 2}. 
\end{proof}
Next, we need to tackle the non-convexity of the HO-aware QoS constraints in $\mathbf{{\bar{C}}_{10}}$. This constraint can be written as:
\begin{align}
   \log\left(\frac{1}{2}R_k^\T + \frac{1}{2}R_k^\U\right) \geq \log\left(\frac{1}{2}R_k^{\mathrm{Th}}\right), \forall \ k.
\end{align}
By seeking Jensen's inequality and setting $\lambda_1 = \lambda_2 = \frac{1}{2}$ in \eqref{eqn:Jensen-Ineq}, we have:
\begin{align}\label{eqn:HO-aware-QoS-App}
    &\log(\frac{1}{2}\bar{R}_k) = \log\left(\frac{1}{2}R_k^\T + \frac{1}{2}R_k^\U\right) \geq \tilde{R}_k\notag \\ & =\frac{1}{2}\log(R_k^\T) + \frac{1}{2}\log(R_k^\U) \geq \log\left(\frac{1}{2}R_k^{\mathrm{Th}}\right)=\tilde{R}_k^{\mathrm{Th}}, \forall \ k. \notag
\end{align}

After the above-mentioned discussions, problem $\hat{\pr}_2$ is approximated with the following convex optimization at the $i$-th iteration of the Algorithm:
\begin{align}
 &\hat{\pr}_3:
\underset{\substack{\alpha^n_{b,k},\w_{b,k},p_{b,k} \\ \beta^n_{r,k},\u_{r,k}, \delta_{r,k}}}{\textrm{maximize}} \,\, \,\,\tilde{U}_2^{(i)} - \Gamma_1 \mathcal{L}_1^{(i)} - \Gamma_2 \mathcal{L}_2^{(i)}\\
&\mbox{s.t.}\hspace*{4mm}
\mathbf{{C}_{1}},\mathbf{{C}_{2}}, \mathbf{{C}_{5.1}}, \mathbf{{C}_{5.2}}, \mathbf{{C}_{5.3}}, \mathbf{{C}_{6.1}}, \mathbf{{C}_{6.2}}, \mathbf{{C}_{6.3}}, \notag \\
&\mathbf{{C}_{7}},\mathbf{{\tilde{C}}_{10}}:\tilde{R}_k \geq \tilde{R}_k^{\mathrm{Th}}, \mathbf{{C}_{11.1}},\mathbf{{C}_{11.3}},\mathbf{{C}_{12}},\mathbf{{C}_{13}},\mathbf{{C}_{14}}, \notag
\end{align}
where $\mathbf{{C}_{13}}$ and $\mathbf{{C}_{14}}$ are included to ensure the conditions of \textbf{Lemma 2}. $\hat{\pr}_3$ is a convex optimization problem that CVX can solve efficiently. Therefore, by properly adjusting \textbf{Algorithm~1}, we can solve $\hat{\pr}_1$, while the convex optimization solved inside the algorithm is $\hat{\pr}_3$.

\subsection{MOOP-based Approach}
In the second HO-aware resource allocation approach, we utilize the same system sum-rate formulation used in $U_1$ and solve the following MOOP problem, which maximizes the system sum-rate \textcolor{black}{as the first objective and minimizes the total number of HOs as the second objective.}
\begin{align}
&\label{prob-bar1} \bar{\pr}_1:
\underset{\substack{\alpha^n_{b,k},\w_{b,k},\F_b \\ \beta^n_{r,k},\u_{r,k},\Q_r}}{\textrm{maximize}} \,\, \,\,U_1 \ \ , \ \ \ \underset{\alpha^n_{b,k},\beta^n_{r,k}}{\textrm{minimize}} \,\, \,\,U_{\mathrm{HO}} \\
&\mbox{s.t.}\hspace*{4mm}
\mathbf{{C}_{1}},\mathbf{{C}_{2}}, \mathbf{{C}_{3}}, \mathbf{{C}_{4}}, \mathbf{{C}_{5}},\mathbf{{C}_{6}}, \mathbf{{C}_{7}}, \mathbf{{C}_{8}}, \mathbf{{C}_{9}}, \mathbf{{{C}}_{10}},
\mathbf{{C}_{11}},\mathbf{{C}_{12}},\notag
\end{align}
where $U_{\mathrm{HO}}$ is the total number of HOs and is given by:
\begin{equation}\label{eqn:NumHOs}
    U_{\mathrm{HO}} = \sum\limits_{k=1}^K \left(\sum\limits_{b=1}^{B^\T}(1-\alpha_{b,k}^{n-1})\alpha_{b,k}^{n} + \sum\limits_{r=1}^{B^\U}(1-\beta_{r,k}^{n-1})\beta_{r,k}^{n}\right).
\end{equation}
To solve $\bar{\pr}_1$, the method of \textit{weighted sum method} is utilized to transform both objective functions into a single one \cite{marler2010weighted}. In this approach, the second objective function $U_{\mathrm{HO}}$ is added to the first objective function with a weight denoted by $\Xi \in \mathbb{R^+}$. Therefore, $\bar{\pr}_1$ is transformed into the following problem:
\begin{align}
&\label{prob-bar2} \bar{\pr}_2:
\underset{\substack{\alpha^n_{b,k},\w_{b,k},\F_b \\ \beta^n_{r,k},\u_{r,k},\Q_r}}{\textrm{maximize}} \,\, \,\,U_1 - \Xi  U_{\mathrm{HO}} \\
&\mbox{s.t.}\hspace*{4mm}
\mathbf{{C}_{1}},\mathbf{{C}_{2}}, \mathbf{{C}_{3}}, \mathbf{{C}_{4}}, \mathbf{{C}_{5}},\mathbf{{C}_{6}}, \mathbf{{C}_{7}}, \mathbf{{C}_{8}}, \mathbf{{C}_{9}}, \mathbf{{{C}}_{10}},
\mathbf{{C}_{11}},\mathbf{{C}_{12}} ,\notag
\end{align}
Problem $\bar{\pr}_2$ can be solved directly using \textbf{Algorithm 1}, which includes obtaining the analog beamformer using \eqref{Analog_SubProb} and the subsequent transformations. 

\section{Numerical Results and Discussions}

In this section, we evaluate the proposed MBN performance and resource allocation methods. 

\subsubsection{Simulation Set-up and Parameters}
We consider a corridor with a length of $D = 350$ [m] and width of 250 [m], where BSs are deployed on both sides uniformly. The minimum vertical distance from the BSs is 30 [m]. Unless otherwise stated, the parameters used in the simulations are as follows: The carrier frequency of THz, $f_\T = 0.4$ [THz], and its bandwidth is $\omega^\T = 0.8$ [GHz]. The carrier frequency at UMB is $f_\U=8$ [GHz] and the bandwidth is set to be $\omega^\U = 100$ [MHz]. The antenna gains are $G_{\mathrm{Tx}}^\T = 15$ [dB], $G_{\mathrm{Rx}}^\T = 8$ [dB], $G_{\mathrm{Tx}}^\U = 10$ [dB], and $G_{\mathrm{Rx}}^\U = 8$ [dB]. Power budgets are $P_{\max}^\T = 25$ [dBm] and $P_{\max}^\U = 40$ [dBm]. Cluster sizes are set as $\mathcal{C}_k^\T=2$ and $\mathcal{C}_k^\U=2$. For scenarios where BSs operate on only one frequency band, the cluster size is set to $\mathcal{C}_k = \mathcal{C}_k^\T + \mathcal{C}_k^\U = 4$ for fairness in comparison. 

\subsubsection{Benchmarks}
The proposed MBN architecture is compared against the \textit{THzOn} scenario, where no UBS is deployed in the corridor, and for a fair comparison, the total number of BSs is the same in both cases. We have \textit{MBN-Algo1}, which denotes the proposed MBN architecture as well as resource allocation in \textbf{Algorithm~1}. The notation FC is used for the fully-connected HBF case, while PC is used for the partially-connected HBF. Meanwhile, \textit{THzOn-Algo1} represents the deployment of only TBSs and using \textbf{Algorithm~1}. \textcolor{black}{\textit{MBN-B1} refers to the case where, in the proposed MBN, the $\mathcal{C}_k^\T$ nearest TBSs and $\mathcal{C}_k^\U$ nearest UBSs with the strongest channel gains are selected for the $k$-th user's association in the THz and UMB bands, respectively. The proposed beamforming method in \textbf{Algorithm 1} is then applied. The same approach is used in \textit{THzOn-B1}, where only TBSs are deployed, and the $\mathcal{C}_k = \mathcal{C}_k^\T + \mathcal{C}_k^\U$ nearest TBSs are selected for user association. Moreover, \textit{MBN-ZF} employs the same nearest user–BS association strategy, but uses regularized zero-forcing (RZF) for beamforming \cite{interdonato2020local}.}
\textcolor{black}{In this method, the digital beamformers at the $b$-th TBS and $r$-th UBS are given by:
    $\bar{\boldsymbol{W}}_b = \boldsymbol{H}_b^H\left(\boldsymbol{H}_b \boldsymbol{H}_b^H + e^\T_b \boldsymbol{I} \right)^{-1}$,
    $\bar{\boldsymbol{U}}_r = \boldsymbol{G}_r^H\left(\boldsymbol{G}_r \boldsymbol{G}_r^H + e^\U_r \boldsymbol{I} \right)^{-1}$,
where $\boldsymbol{H}_b = \boldsymbol{F}_b^H [\alpha_{b,1}^n\h_{b,1},\dots,\alpha_{b,K}^n\h_{b,K}]$, with $e_b^\T >0$ used to regularize $\boldsymbol{H}_b \boldsymbol{H}_b^H$, and $\boldsymbol{G}_r = \boldsymbol{Q}_r^H [\beta_{r,1}^n\g_{r,1},\dots,\beta_{r,K}^n\g_{r,K}]$, with $e_r^\U >0$ used to regularize $\boldsymbol{G}_r \boldsymbol{G}_r^H$. To satisfy the power budget constraints, the final beamformers are normalized as follows: ${\boldsymbol{W}}_b = \bar{\boldsymbol{W}}_b \sqrt{\frac{\bar{P}_{\max}^\T}{\sum_{k=1}^{K}\alpha_{b,k}^n\mathrm{tr}(\bar{\boldsymbol{W}}_b \bar{\boldsymbol{W}}^H_b)}}$, and ${\boldsymbol{U}}_r = \bar{\boldsymbol{U}}_r \sqrt{\frac{\bar{P}_{\max}^\U}{\sum_{k=1}^{K}\beta_{r,k}^n \mathrm{tr}(\bar{\boldsymbol{U}}_r \bar{\boldsymbol{U}}^H_r)}}$. A similar approach is also taken for \textit{THzOn-ZF}.
}

\subsubsection{\textcolor{black}{Convergence of Algorithm 1}}
\begin{figure}
    \centering
\includegraphics[scale=0.5]{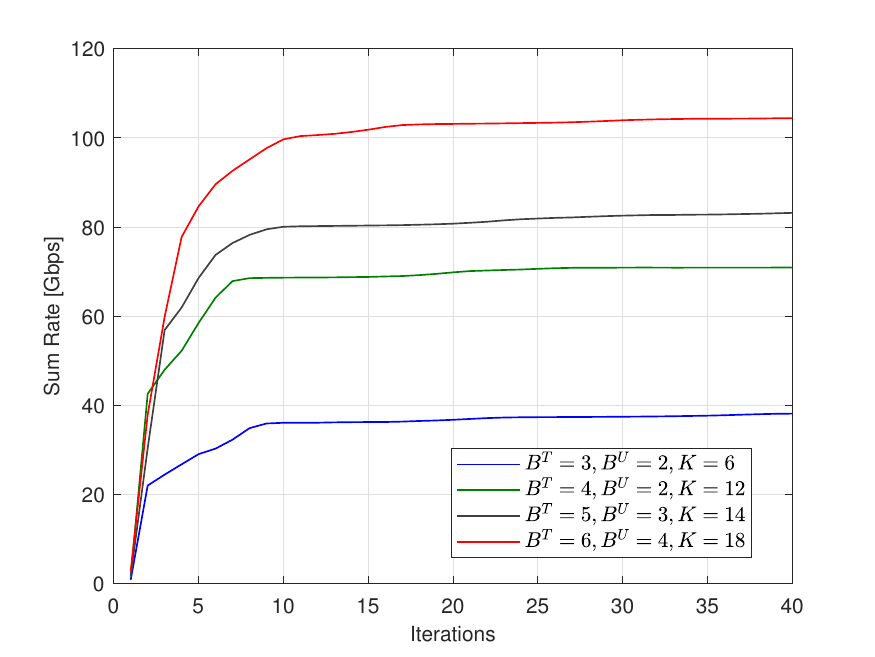}
    \caption{\textcolor{black}{Convergence of Algorithm 1 for different system parameters}}
    \label{fig:Convergence}
\end{figure}

\textcolor{black}{Figure \ref{fig:Convergence} illustrates the convergence behavior of the proposed algorithm under different system parameters, including the number of users and BSs. It can be observed that \textbf{Algorithm 1} converges, as mathematically shown in Section III.E.1. As the number of variables increases, more iterations are required for convergence due to the expanded search space and increased complexity of the optimization problem in each iteration.
}

\subsection{Numerical Results for Stationary Users}
\subsubsection{Impact of Molecular Absorption Coefficient}
\begin{figure}
    \centering
\includegraphics[scale=0.5]{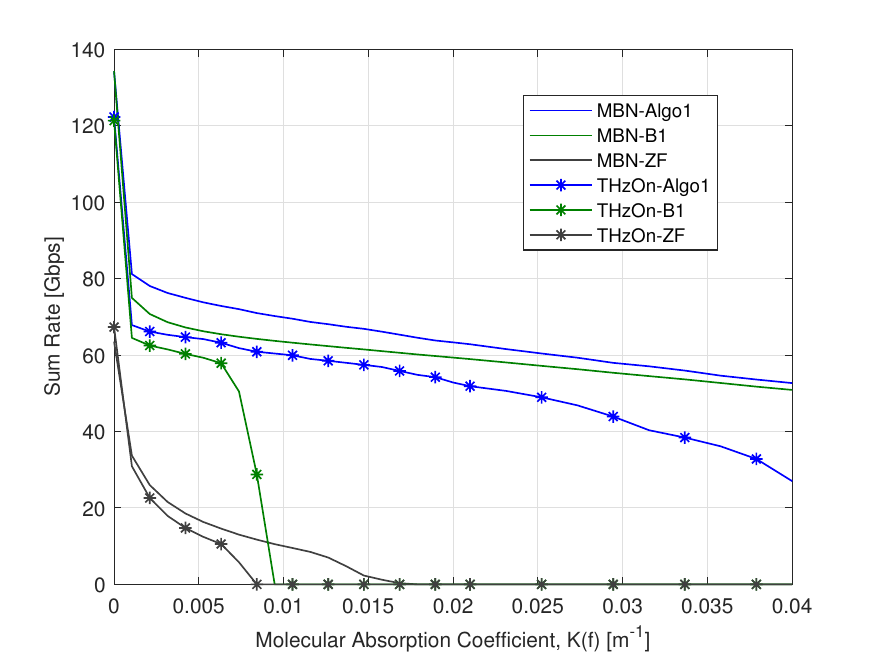}
    \caption{Network  sum-rate at different molecular absorption coefficients, $R^{\mathrm{Th}} = 0.5$ [Gbps], $K=12$, $B^\T = 4$, $B^\U = 2$, $B^\T_{\mathrm{THzOn}} = 6$, $M^\T=504$, $M^\U=84$}
    \label{fig:AbsorbCoef}
\end{figure}
Different environmental conditions, including humidity, temperature, and gas molecule ratios, can impact the molecular absorption coefficient, thereby affecting transmission performance over the THz frequency band. Varying the molecular absorption coefficients for a given carrier frequency can provide insights related to different environmental conditions \cite{MBN-Survey,MBN-Magazine}.
Figure~\ref{fig:AbsorbCoef} illustrates the impact of increasing molecular absorption coefficients on the proposed MBN network, THzOn, and their respective benchmarks. As shown, the performance of all approaches declines with increasing $\mathcal{K}(f_\T)$. \textcolor{black}{This is because the molecular absorption loss in the direct channels $\h_{b,k}$ increases significantly, and the molecular absorption noise further degrades the SINR in the THz band.} The sum-rate for the proposed MBN network decreases at a slower rate, outperforming the scenario where only THz BSs are deployed. Also, the reason MBN-B1 approaches the performance of MBN-Algo1 is that a higher absorption coefficient leads to greater absorption loss, which increases with distance. When the molecular absorption loss becomes dominant, the optimal user association determined by Algorithm 1 converges to the $\mathcal{C}_k$ nearest BSs association to user $k$. 
Additionally, the proposed algorithm exhibits better performance compared to other benchmarks. Specifically, the sum-rate of THzOn-B1 rapidly approaches zero, while the proposed method maintains superior performance. Moreover, RZF in the MBN also outperforms THzOn-ZF, though both approaches see their sum-rates decline more quickly compared to the proposed Algorithm 1.

\subsubsection{Impact of Blockage Density}
\begin{figure}
    \centering
\includegraphics[scale=0.5]{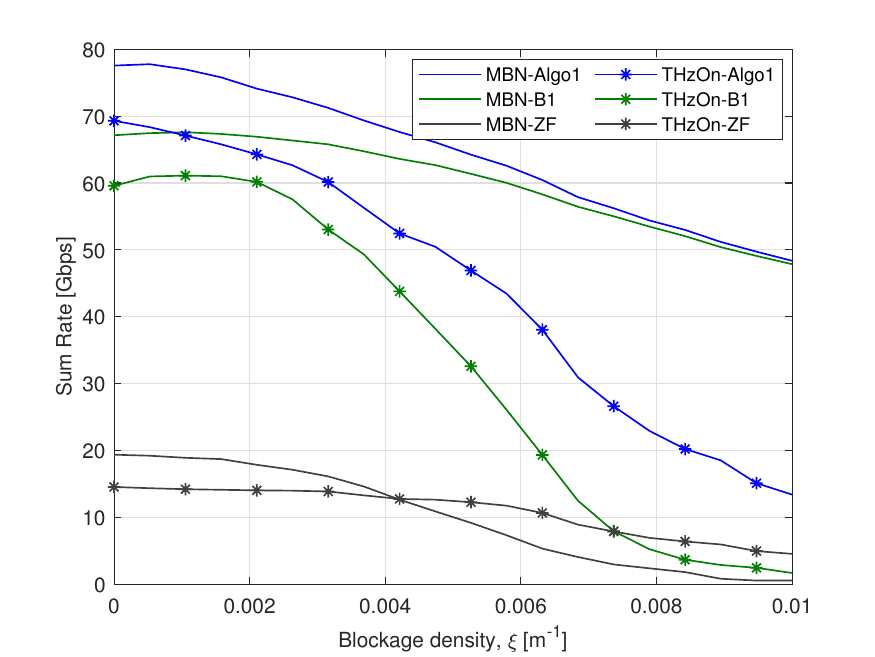}
    \caption{System sum-rate versus density of blockers, $\xi$. $R^{\mathrm{Th}} = 0.5$ [Gbps], $K=12$, $B^\T = 4$, $B^\U = 2$, $B^\T_{\mathrm{THzOn}} = 6$, $M^\T=504$, $M^\U=84$}
    \label{fig:Blockage}
\end{figure}

Figure \ref{fig:Blockage} illustrates the impact of increasing the density of blockers. As shown, the proposed MBN outperforms THzOn as it can benefit from the availability of UMB frequency as blockage density increases. Additionally, the proposed \textbf{Algorithm~1} outperforms other benchmarks. At high blockage densities, the performance of MBN-B1 converges to that of the proposed algorithm because the optimal association in MBN-Algo1 tends to favor non-blocked links, which are typically the nearest BSs to the user, as shown in \eqref{eqn:Blockage-Dist}.
Also, after a certain point, THzOn-ZF outperforms MBN-ZF, which is attributed to the non-optimal resource allocation in this scheme.

\subsubsection{Significance of MBN Deployment with Fully-Connected and Partially-Connected HBF}
\begin{figure}
    \centering
\includegraphics[scale=0.52]{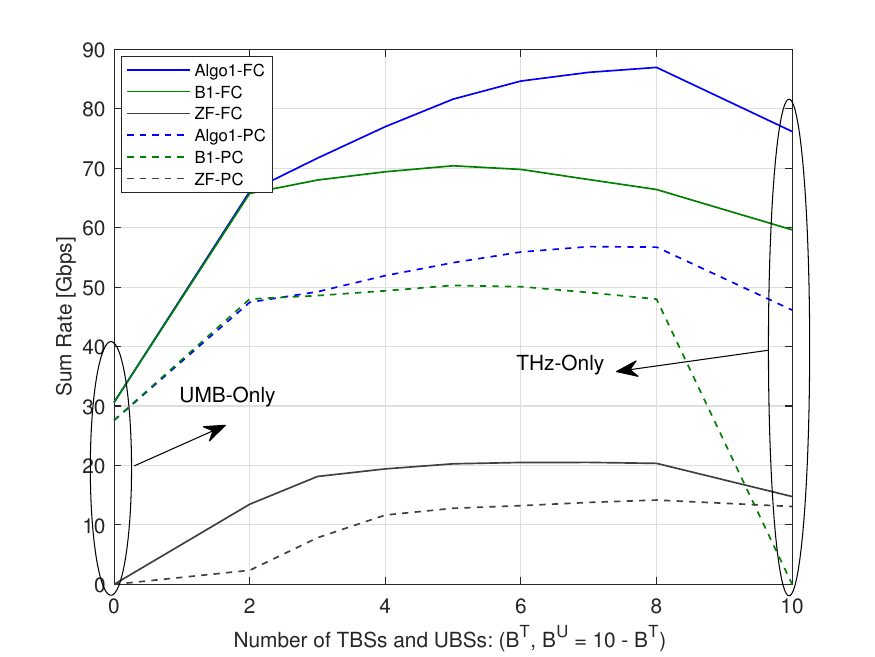}
    \caption{System sum-rate versus number of TBSs and UBSs. The total number of BSs is 10 and increasing TBSs decreases UBSs as $B^\U$ varies within the range of $[10,0]$. $R^{\mathrm{Th}} = 0.5$ [Gbps], $K=12$, $M^\T=504$, $M^\U=84$.}
    \label{fig:NumBSs-Varies}
\end{figure}

Figure \ref{fig:NumBSs-Varies} illustrates the impact of deploying different types of BSs. We assume a total of 10 BSs, and as the number of TBSs on the x-axis increases, the number of UBSs decreases accordingly. For example, when the x-axis value is $B^\T = 6$, the number of UBSs is $B^\U = 4$. 
\textcolor{black}{In this figure, all of the channel characteristics, including blockage with $\xi = 0.002$, are considered in the transitioning from UMB-only to MBN and then to THz-only network.}
It is evident that the MBN can outperform both UMB-only and THz-Only setups for a predefined deployment ratio. The figure also demonstrates the superiority of the proposed algorithm over other benchmarks.
Moreover, in both the proposed Algorithm~1 and Benchmark~B1, where beamforming is optimized, FC HBF consistently outperforms PC HBF due to utilizing more phase shifters at the cost of higher energy consumption. However, it can also be observed that even with PC HBF, the MBN still outperforms the single-band counterparts.

\subsubsection{Number of Antennas and NFC vs FFC}
\begin{figure}
    \centering
\includegraphics[scale=0.5]{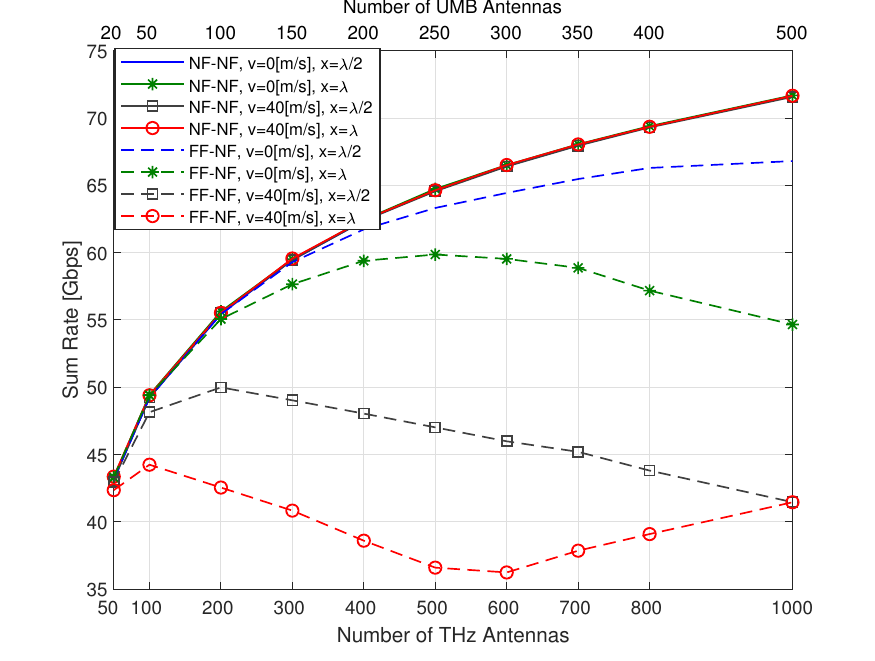}
    \caption{\textcolor{black}{System sum-rate versus the different number of antennas of UMB and THz bands for various velocities of users and antenna spacing. $R^{\mathrm{Th}} = 0.1$ [Gbps], $K=8$, $B^\T = 4$, $B^\U = 2$, $\mathcal{C}_k^\T=4$, $\mathcal{C}_k^\U=2$}}
    \label{fig:NumAnt-Varies}
      \vspace{-2mm}
\end{figure}
\textcolor{black}{Figure \ref{fig:NumAnt-Varies} illustrates the impact of increasing the number of antennas in both frequency bands on the system sum-rate across different user velocities and antenna spacing. 
In the cases labeled `NF-NF,' resource allocation is performed using \textbf{Algorithm-1} based on NFC channel modeling, whereas in `FF-NF,' resource allocation is carried out using FFC channel modeling, with NFC considered as the ground truth.
Increasing both the number of antennas and the antenna spacing expands the array aperture, which in turn increases the Fraunhofer distance. Below this distance, NFC provides a more accurate representation compared to FFC. It can be observed that increasing antenna spacing amplifies the sensitivity to the choice of channel modeling. Furthermore, in the presence of mobility, when the antenna spacing is half a wavelength, i.e., $x = \frac{\lambda}{2}$, increasing user velocity further degrades system performance due to mismatched beamforming under the FFC assumption. However, it can also be observed that when $x = \lambda$, the performance degradation levels off at a certain point, which may be attributed to the energy spread effect \cite{ELAA-NFC-1}. It is also worth noting that employing near-field modeling increases the complexity of channel acquisition in practice, as discussed in \cite{deshpande2022wideband}. This introduces a trade-off between the performance improvement offered by NFC-based beamforming and the additional overhead required for CSI acquisition.}
\subsection{Numerical Results for Mobile Users}
In the simulation of mobile users moving at a speed of 40 [m/s], we consider \textcolor{black}{three consecutive optimization windows (trajectory points), each of length 100 [ms] \cite{gupta2024forecaster}}, and report the average HO-aware sum-rate in \eqref{eqn:HO-aware-SumRate} and the number of HOs in \eqref{eqn:NumHOs}. The remaining parameters are specified in the description of each simulation result.
In these results, \textit{MBN-Algo1-Cost} represents the first proposed HO-aware resource allocation, while \textit{MBN-Algo1-MO} refers to the second HO-aware method. Similar approaches are applied to the THz-Only case. \textcolor{black}{For benchmarks labeled “(No Cost),” resource allocation is performed via \textbf{Algorithm 1} without HO consideration (i.e., HO-agnostic), while their achievable rates are computed using the HO-aware rate expressions in \eqref{eqn:THz-HO-Rate} and \eqref{eqn:RF-HO-Rate}.}

\subsubsection{\textcolor{black}{Impact of Cluster Size}}
\textcolor{black}{
Figure \ref{fig:ClusterSize-Varies} illustrates the trade-off between cluster size and the number of HOs when no HO control is applied. For the THzOn case, we consider two configurations: one with a cluster size of $\mathcal{C}_k = \mathcal{C}_k^\T + \mathcal{C}_k^\U$, as used in the other results of this work, and one with $\mathcal{C}_k = \mathcal{C}_k^\T$. It is evident that increasing the cluster size enhances the system sum-rate due to greater cooperative beamforming capability. On the other hand, the total number of HOs initially increases with cluster size and then decreases. While more HOs can degrade system performance, the improvement from larger cluster sizes can offset this effect.
As shown, when the cluster size increases from $\mathcal{C}_k = 1 + 1$ to $\mathcal{C}_k = 2 + 2$, the green curve for THzOn experiences a performance drop, followed by improvement beyond that point. A similar trend is observed for the black curve of THzOn, where the performance decreases from $\mathcal{C}_k = 3$ to $\mathcal{C}_k = 5$, coinciding with the highest number of HOs.
It is also observed that MBN results in fewer HOs and consequently achieves a higher HO-aware sum-rate.
}
\begin{figure}
    \centering
\includegraphics[scale=0.5]{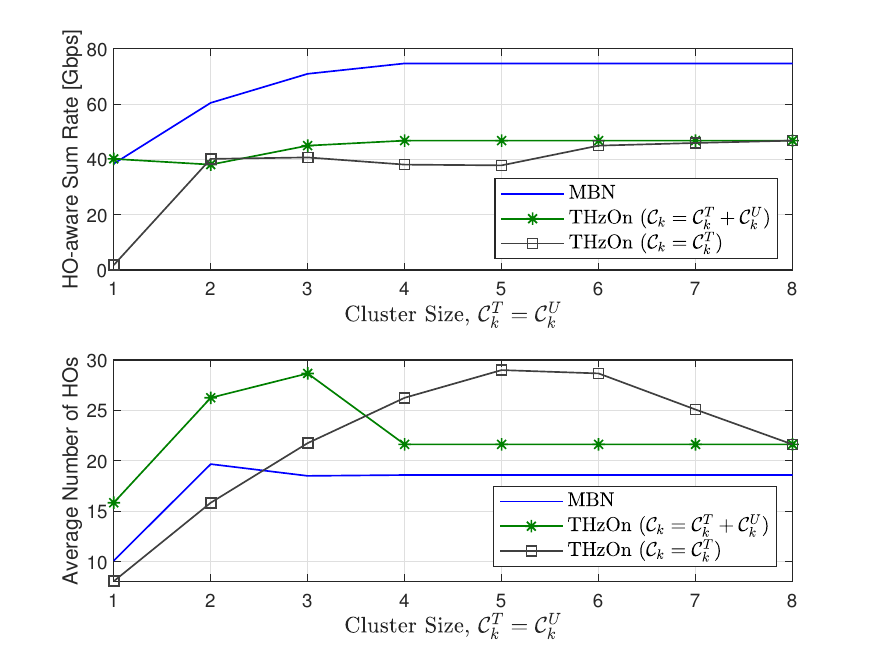}
    \caption{\textcolor{black}{HO-aware sum-rate and the total number of HOs versus different cluster size $(\mathcal{C}_k^\T=\mathcal{C}_k^\U)$, where $R^{\mathrm{Th}} = 0.5$ [Gbps], $K=12$, $B^\T = 5$, $B^\U = 3$, $B^\T_{\mathrm{THzOn}} = 8$, and $\eta^\T=\eta^\U = 0.4$.}}
    \label{fig:ClusterSize-Varies}
\end{figure}

\subsubsection{Impact of HO Cost}

In Fig.~\ref{fig:Sumrate-EtaHO-Varies}, the impact of increasing the HO cost for both frequency bands on the proposed methods is illustrated. First, note that the proposed MBN without HO consideration in the optimization, referred to as MBN-Algo1(No Cost), outperforms its counterpart with only THz BSs deployed, i.e., THzOn-Algo1(No Cost). The THzOn scenario is observed to be more sensitive to increasing HO costs. Thus, even without accounting for HOs during resource allocation, the proposed MBN network outperforms its single-band counterpart.
While the performance of MBN-Algo1-MO does not drop as the HO cost increases, the main challenge with this method is determining an optimal value for $\Xi$, which serves as a hyperparameter and varies across scenarios
Moreover, while MBN-B1(No Cost) is more robust to HOs, it lacks the flexibility to optimize user association and mitigate the severe effects of HOs. The proposed HO-aware methods also improve the performance of the THzOn case, although it underperforms its MBN counterpart as the HO cost increases.

\begin{figure}
    \centering
\includegraphics[scale=0.5]{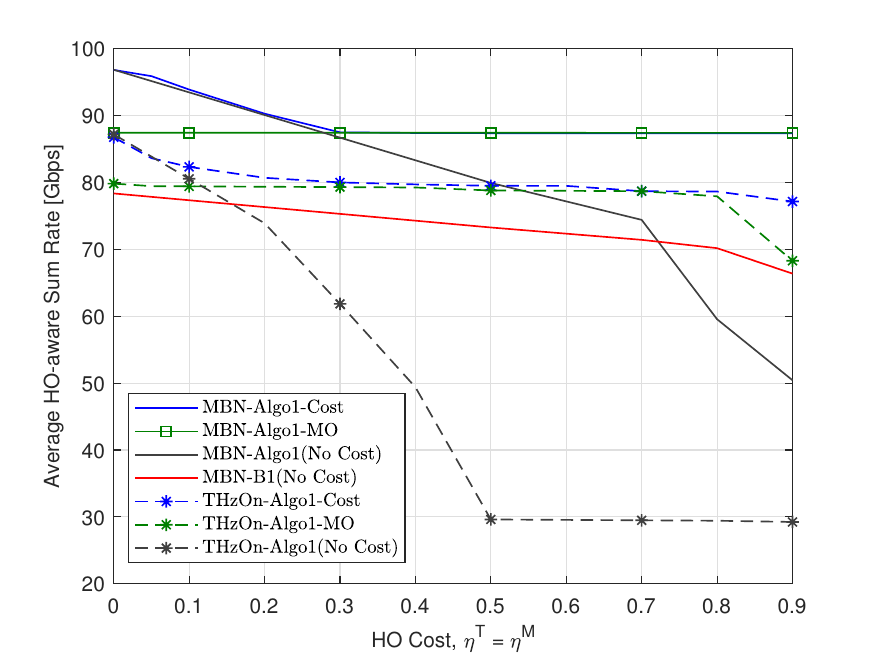}
    \caption{\textcolor{black}{HO-aware sum-rate versus different HO cost $(\eta^\T=\eta^\U)$, where $R^{\mathrm{Th}} = 0.5$ [Gbps], $K=15$, $B^\T = 5$, $B^\U = 3$, $B^\T_{\mathrm{THzOn}} = 8$, and $\Xi=1$.}}
    \label{fig:Sumrate-EtaHO-Varies}
      \vspace{-2mm}
\end{figure}

Fig.~\ref{fig:NumHOs-EtaHO-Varies} shows the impact of increasing the HO cost on the number of HOs. As seen, increasing the HO cost significantly reduces the number of HOs in MBN-Algo1-Cost for both MBN and THzOn scenarios. Also, the MOOP-based method maintains a nearly constant number of HOs, which is due to the absence of a direct relationship between HO and rate in its optimization formulation. MBN-B1(No Cost) results in fewer HOs compared to MBN-Algo1(No Cost), although it lacks the capability to achieve optimal resource allocation. Moreover, we note that even without HO consideration in the resource allocation, MBN imposes a smaller number of HOs in mobile scenarios.

\begin{figure}
    \centering
\includegraphics[scale=0.5]{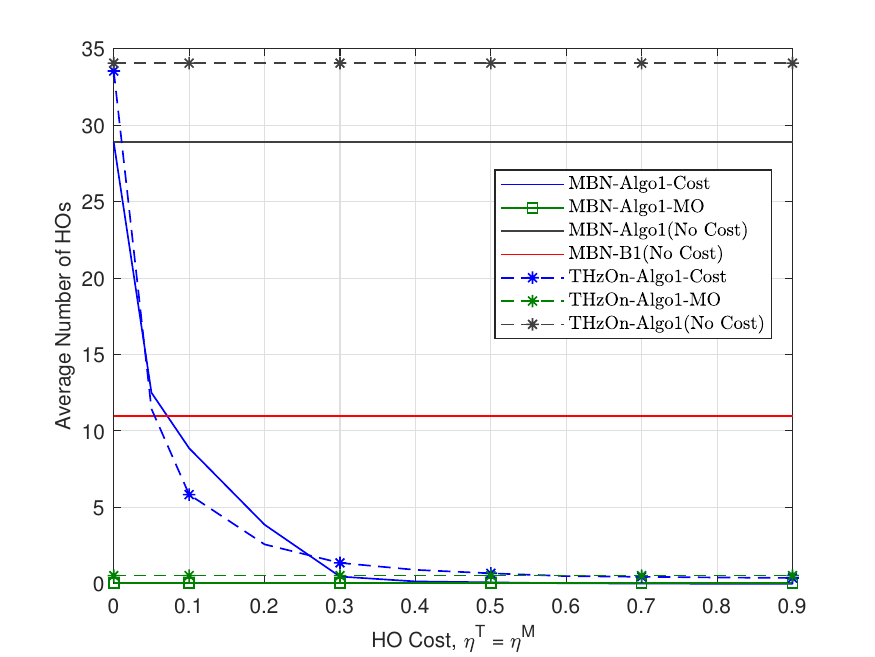}
    \caption{\textcolor{black}{Total number of HOs versus different HO cost $(\eta^\T=\eta^\U)$, where $R^{\mathrm{Th}} = 0.5$ [Gbps], $K=15$, $B^\T = 5$, $B^\U = 3$, $B^\T_{\mathrm{THzOn}} = 8$, and $\Xi=1$.}}
    \label{fig:NumHOs-EtaHO-Varies}
\end{figure}

\subsubsection{Impact of Minimum Rate Requirement}

Fig.~\ref{fig:SumRate_Rth_Varies} illustrates the impact of increasing the minimum rate requirement on the HO-aware system sum-rate. The proposed HO-aware resource allocation consistently outperforms the other methods. Additionally, due to a higher number of HOs in the THzOn system compared to MBN, THzOn demonstrates inferior performance. It is also observed that increasing the minimum rate requirement has a more adverse effect on THzOn-Algo1-MO than on THzOn-Algo1-Cost, highlighting the importance of incorporating HO-aware QoS constraints. 

\begin{figure}
    \centering
\includegraphics[scale=0.5]{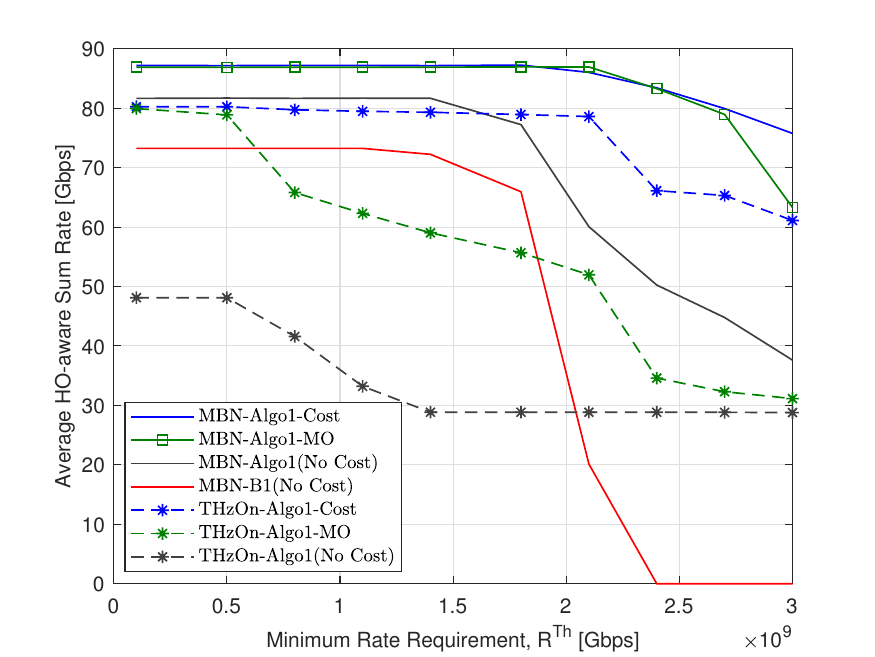}
    \caption{\textcolor{black}{HO-aware  sum-rate versus minimum rate requirement $R^{\mathrm{Th}}$, where $K=15$, $B^\T = 5$, $B^\U = 3$, $B^\T_{\mathrm{THzOn}} = 8$, $\Xi=1$, and $\eta^\T=\eta^\U = 0.4$}}
    \label{fig:SumRate_Rth_Varies}
    \vspace{-2mm}
\end{figure}

\subsubsection{\textcolor{black}{Impact of Imperfect CSI}}
\begin{figure}
    \centering
\includegraphics[scale=0.5]{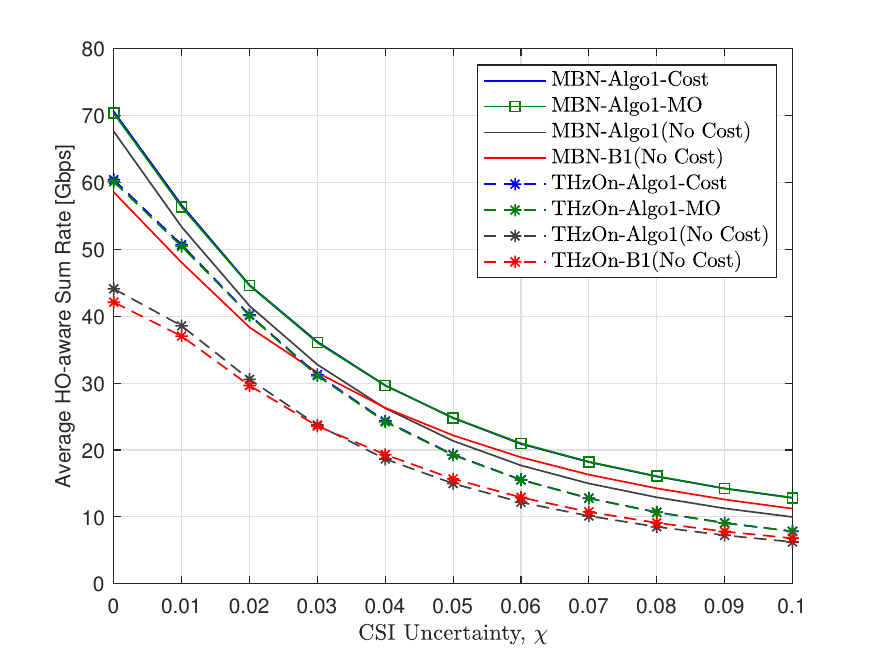}
    \caption{\textcolor{black}{HO-aware system sum-rate versus CSI uncertainty $\chi$, where $K=12$, $B^\T = 4$, $B^\U = 2$, $B^\T_{\mathrm{THzOn}} = 6$, $\Xi=1$, and $\eta^\T=\eta^\U = 0.4$}}
    \label{fig:ImpCSI}
    \vspace{-2mm}
\end{figure}
\textcolor{black}{
Since obtaining perfect CSI is challenging in practice, we examine the robustness of the proposed methods to CSI uncertainties by considering a statistical CSI error model, as in \cite{ImpCSI-1}. For instance, for THz links, the channel is expressed as:
\begin{equation}\label{eqn:ImpCSI}
\h_{b,k}^{\text{Ideal}} = \h_{b,k} + \Delta \h_{b,k}, \quad \forall b,k,
\end{equation}
where $\h_{b,k}^{\text{Ideal}}$ denotes the perfect CSI, $\h_{b,k}$ is the estimated CSI, and $\Delta \h_{b,k}$ represents the CSI estimation error. The error term $\Delta \h_{b,k}$ is modeled as a circularly symmetric complex Gaussian vector.
Following \cite{ImpCSI-1,ImpCSI-2}, the error variance is defined as $\hat{\chi} = \chi {||\h_{b,k}||}^2_2$, where $\chi \in [0,1)$ quantifies the level of CSI uncertainty. Assuming the same value of $\chi$ for all channels, Figure~\ref{fig:ImpCSI} illustrates the impact of imperfect CSI on the proposed methods. As the CSI uncertainty increases, the system performance of all methods degrades. This degradation is attributed to the large number of antennas in ELAAs, which amplifies the CSI error variance as described in \eqref{eqn:ImpCSI}.
It can be observed that for both MBN and THz-On, the performance of MBN-B1 and THzOn-B1 approaches and outperforms that of the proposed HO-unaware methods as the CSI uncertainty increases. This is because, when the estimation error is proportional to the channel gain, resource allocation strategies based on associating with the nearest BSs tend to be less sensitive to CSI errors. Nevertheless, the proposed HO-aware methods still outperform these benchmarks due to their explicit consideration of HOs in the objective function.
}

\section{Conclusion}
This paper presents a comprehensive sum-rate maximization framework for cooperative UMB/THz MBNs through the joint optimization of hybrid beamforming and user association, subject to constraints such as power budget, QoS, and maximum cluster size, for both stationary and moving users. The proposed framework incorporates NFC channel modeling, fully and partially connected hybrid beamforming architectures, and users' mobility.  
\textcolor{black}{Specifically, the coexistence of UMB/THz in the considered MBN improves the system performance under higher molecular absorption as well as THz link blockage. While near-field channel modeling requires more parameters to be known in practice, performing NFC-based beamforming, as opposed to mismatched FFC-based beamforming, becomes essential when using ELAAs. Moreover, we observe that increasing the cluster size in cooperative systems is not always beneficial in mobility scenarios due to HO costs. Additionally, the proposed HO-aware resource allocation methods effectively mitigate the impact of HOs, particularly when the cost of each HO and the minimum rate requirement of each user are high.  Moreover, investigating the integrated deployment of MBN, as proposed in \cite{MBN-Survey}, within the considered framework is a potential research topic. Exploring the feasibility of ML-based methods for solving the optimization problem, such as the approach in \cite{alizadeh2023power}, is another direction for future research.}


\section*{Appendix}

\subsection{Proof of \textbf{Lemma 1}}
\label{Appendix-A}
Given the optimal and feasible $\boldsymbol{X}^o$ in \eqref{Lem1-P1}, define $y_{i,j}^o=\frac{\sqrt{A_{i,j}(\boldsymbol{X}^o)}}{B_{i,j}(\boldsymbol{X}^o)}$. Then, it can be shown that the term $i$ and $j$ in the objective function in \eqref{Lem1-P2} at $y_{i,j}^o$ is: $T_{i,j}(\boldsymbol{X}^o,y_{i,j}^o) = \frac{A_{i,j}(\boldsymbol{X}^o)}{B_{i,j}(\boldsymbol{X}^o)}$. Hence, by taking the logarithm, we have:
{\small
\begin{align}
    & \log\hspace*{-0.1mm}\left(\hspace*{-0.1mm}2\Re\left\{\hspace*{-0.1mm}y_{i,j}^{o,*}\sqrt{A_{i,j}(\boldsymbol{X}^o)}\hspace*{-0.1mm}\right\}\hspace*{-0.5mm}-\hspace*{-0.5mm}\abss{y_{i,j}^o}\hspace*{-0.5mm}B_{i,j}(\boldsymbol{X}^o)\hspace*{-0.1mm}\right)\hspace*{-0.1mm} =\hspace*{-0.1mm} \log\hspace*{-0.1mm}\left(\hspace*{-0.1mm}\frac{A_{i,j}(\boldsymbol{X}^o)}{B_{i,j}(\boldsymbol{X}^o)}\hspace*{-0.1mm}\right)\notag
\end{align}
}
and by summing over all $i$ and $j$, the objective function in \eqref{Lem1-P2} at $(\boldsymbol{X}^o,\boldsymbol{y}^o)$ matches the objective of the problem in \eqref{Lem1-P1}.

Also, since $\boldsymbol{X}^o$ is feasible in \eqref{Lem1-P1}, we have: $\sum\limits_{j} \log\left(\frac{A_{i,j}(\boldsymbol{X}^o)}{B_{i,j}(\boldsymbol{X}^o)}\right) \geq R^{Th}_i$, and by using the result of the previous step and evaluating the terms in $(\boldsymbol{X}^o,\boldsymbol{y}^o)$, we have:
{\small
\begin{align}
    & \sum\limits_{j} \log\hspace*{-0.1mm}\left(\hspace*{-0.1mm}2\Re\left\{\hspace*{-0.5mm}y_{i,j}^{o,*}\sqrt{A_{i,j}(\boldsymbol{X}^o)}\hspace*{-0.5mm}\right\}\hspace*{-0.1mm}-\hspace*{-0.1mm}\abss{y^o_{i,j}}B_{i,j}(\boldsymbol{X}^o)\right) \geq R^{Th}_i, \forall i.\notag
\end{align}
}
Therefore, the objective and feasible set are equivalent for $\boldsymbol{y}^o$. 
Now for optimality in the second problem, suppose there exists a feasible solution $(\boldsymbol{X}^{\prime},\boldsymbol{y}^{\prime})$ to \eqref{Lem1-P2} such that $\sum_{i} \sum_{j} \log \left(2 \Re\left\{y_{i, j}^{\prime *} \sqrt{A_{i, j}\left(\boldsymbol{X}^{\prime}\right)}\right\}-\left|y_{i, j}^{\prime}\right|^2 B_{i, j}\left(\boldsymbol{X}^{\prime}\right)\right)  >\sum_{i} \sum_{j} \log \left(\frac{A_{i, j}\left(\boldsymbol{X}^o\right)}{B_{i, j}\left(\boldsymbol{X}^o\right)}\right).$
Since the maximum of the term inside the logarithm occurs at $y_{i,j}= \frac{\sqrt{A_{i,j}(\boldsymbol{X})}}{B_{i,j}(\boldsymbol{X})}$, we have:
{\small
\begin{align}\label{Lem1-Eq1}
    2 \Re\left\{y_{i, j}^{\prime *} \sqrt{A_{i, j}\left(\boldsymbol{X}^{\prime}\right)}\right\}-\left|y_{i, j}^{\prime}\right|^2 B_{i, j}\left(\boldsymbol{X}^{\prime}\right) \leq \frac{A_{i, j}\left(\boldsymbol{X}^{\prime}\right)}{B_{i, j}\left(\boldsymbol{X}^{\prime}\right)}.
\end{align}
}
Moreover, as the logarithm is a monotonically increasing function, by taking the logarithm and summing over $j$ and $i$ on \eqref{Lem1-Eq1}, we have:
{\small
\begin{align}
    & \sum_{i} \sum_{j} \log \left(2 \Re\left\{y_{i, j}^{\prime *} \sqrt{A_{i, j}\left(\boldsymbol{X}^{\prime}\right)}\right\}-\left|y_{i, j}^{\prime}\right|^2 B_{i, j}\left(\boldsymbol{X}^{\prime}\right)\right) \notag \\ & \leq \sum_{i} \sum_{j} \log \left(\frac{A_{i, j}\left(\boldsymbol{X}^{\prime}\right)}{B_{i, j}\left(\boldsymbol{X}^{\prime}\right)}\right),\notag
\end{align}
}
which implies:
{\small
\begin{align}
    \sum_{i} \sum_{j} \log \left(\frac{A_{i, j}\left(\boldsymbol{X}^{\prime}\right)}{B_{i, j}\left(\boldsymbol{X}^{\prime}\right)}\right)>\sum_{i} \sum_{j} \log \left(\frac{A_{i, j}\left(\boldsymbol{X}^o\right)}{B_{i, j}\left(\boldsymbol{X}^o\right)}\right),\notag
\end{align}
}
which contradicts the optimality of $\boldsymbol{X}^o$ in the first problem in \eqref{Lem1-P1}. Therefore, $\boldsymbol{X}^o$ is the optimal solution to the first problem while $(\boldsymbol{X}^o,\boldsymbol{y}^o)$ is optimal for the second problem.

    
\subsection{Proof of \textbf{Lemma 2}}
\label{Appendix-B}
Each HO-aware rate expression after applying the quadratic transformation can be given in the format of $f(x,y) = (1-y)\log(1+x-x^2)$. To show that $f(x,y)$ is a log-concave function, we need to show that the Hessian of $g(x,y)=\log(f(x,y))$ is negative-definite and $f(x,y) > 0$. The Hessian of $g(x,y)$ is given by:
\begin{align}
    \boldsymbol{H} = \nabla^2 g  = \begin{bmatrix}
\frac{(-2x^2 + 2x - 3) \log(-x^2 + x + 1) - (1 - 2x)^2}{(-x^2 + x + 1)^2 \log^2(-x^2 + x + 1)}
 & 0 \\
0 & -\frac{1}{(1 - y)^2}
\end{bmatrix} \notag
\end{align}
We need to show that the determinant of $\boldsymbol{H}$ is non-negative, which implies that the leading element of $\boldsymbol{H}$ must be non-positive. It can be concluded that the leading element $\boldsymbol{H}_{11}(x) \leq 0, \ \ \forall x\in(0,1)$. This region also complies with the condition of $x-x^2>0$, which is necessary to avoid negative SINR. Therefore, both of the eigenvalues of $\boldsymbol{H}$ are non-positive and the Hessian of $g(x,y)$ is negative semidefinite. On the other hand, with the condition of $f(x,y) > 0$ and the fact that HO cost cannot be negative, we have $ 0 \leq y < 1$ and $x-x^2 > 1$, where the latter is already met. Therefore, we can conclude that for $y \in [0,1)$ and $x\in(0,1)$, the function $f(x,y)$ is log-concave. Also, the condition of $y \in [0,1)$ results in $0 \leq \sum\limits_{b=1}^{B^\T}\left(1-\alpha_{b,k}^{n-1}\right)\alpha_{b,k}^{n} < \frac{1}{\eta^\T}$ for THz and $ 0\leq\sum\limits_{r=1}^{B^\U}\left(1-\beta_{r,k}^{n-1}\right)\beta_{r,k}^{n}<\frac{1}{\eta^\U}$ for UMB.

\bibliography{ref}

\begin{thebibliography}{10}
\providecommand{\url}[1]{#1}
\csname url@samestyle\endcsname
\providecommand{\newblock}{\relax}
\providecommand{\bibinfo}[2]{#2}
\providecommand{\BIBentrySTDinterwordspacing}{\spaceskip=0pt\relax}
\providecommand{\BIBentryALTinterwordstretchfactor}{4}
\providecommand{\BIBentryALTinterwordspacing}{\spaceskip=\fontdimen2\font plus
\BIBentryALTinterwordstretchfactor\fontdimen3\font minus \fontdimen4\font\relax}
\providecommand{\BIBforeignlanguage}[2]{{%
\expandafter\ifx\csname l@#1\endcsname\relax
\typeout{** WARNING: IEEEtran.bst: No hyphenation pattern has been}%
\typeout{** loaded for the language `#1'. Using the pattern for}%
\typeout{** the default language instead.}%
\else
\language=\csname l@#1\endcsname
\fi
#2}}
\providecommand{\BIBdecl}{\relax}
\BIBdecl

\bibitem{Mid-Band-1}
S.~Kang, M.~Mezzavilla, S.~Rangan, A.~Madanayake, S.~B. Venkatakrishnan, G.~Hellbourg, M.~Ghosh, H.~Rahmani, and A.~Dhananjay, ``Cellular wireless networks in the upper mid‑band,'' \emph{IEEE Open J. Commun. Soc.}, vol.~5, no.~0, pp. 2058--2075, Feb. 2024.

\bibitem{Mid-Band-2}
E.~Bj{\"o}rnson, F.~Kara, N.~Kolomvakis, A.~Kosasih, P.~Ramezani, and M.~B. Salman, ``Enabling {6G} performance in the upper mid‑band by transitioning from massive to gigantic {MIMO},'' \emph{IEEE Open J. Commun. Soc.}, 2025, preprint available on arXiv:2407.05630, Jul. 8, 2024.

\bibitem{saeidi2024molecularabsorptionaware}
M.~Amin~Saeidi, H.~Tabassum, and M.~Alizadeh, ``Molecular absorption-aware user assignment, spectrum, and power allocation in dense {THz} networks with multi-connectivity,'' \emph{IEEE Trans. Wirel. Commun.}, vol.~23, no.~11, pp. 16\,404--16\,420, 2024.

\bibitem{itu_report_imt_above_100ghz}
\BIBentryALTinterwordspacing
A.~Weissberger. (2023, Jul.) New {ITU} report in progress: Technical feasibility of {IMT} in bands above 100 {GHz} (92 {GHz} and 400 {GHz}). TechBlog - IEEE Communications Society. Blog post. [Online]. Available: \url{https://techblog.comsoc.org/category/imt-above-100-ghz/}
\BIBentrySTDinterwordspacing

\bibitem{MBN-Survey}
S.~B. Aboagye, M.~A. Saeidi, H.~Tabassum, Y.~Tayyar, E.~Hossain, H.-C. Yang, and M.-S. Alouini, ``Multi-band wireless communication networks: Fundamentals, challenges, and resource allocation,'' \emph{IEEE Trans. Commun.}, vol.~72, no.~7, pp. 4333--4352, Jul. 2024.

\bibitem{MBN-Magazine}
M.~A. Saeidi, H.~Tabassum, and M.-S. Alouini, ``Multi-band wireless networks: Architectures, challenges, and comparative analysis,'' \emph{IEEE Commun. Magazine}, vol.~62, no.~1, pp. 80--86, 2024.

\bibitem{yuan2023joint}
X.~Yuan, F.~Tang, M.~Zhao, and N.~Kato, ``Joint rate and coverage optimization for the {THz/RF} multi‑band communications of space‑air‑ground integrated network in {6G},'' \emph{IEEE Trans. Wirel. Commun.}, vol.~23, no.~6, pp. 6669--6682, Jun. 2024.

\bibitem{MBN-mmWave-6GHz}
C.~C. et~al., ``Deep reinforcement learning for resource allocation in multi-band and hybrid {OMA-NOMA} wireless networks,'' \emph{IEEE Trans. Commun.}, vol.~71, no.~1, pp. 187--198, Jan. 2023.

\bibitem{MBN-ZJ}
Z.~Yan and H.~Tabassum, ``Generalized multi-objective reinforcement learning with envelope updates in {URLLC}-enabled vehicular networks,'' \emph{IEEE Trans. Veh. Technol.}, pp. 1--17, 2025.

\bibitem{NFC-Tutorial}
Y.~Liu, Z.~Wang, J.~Xu, C.~Ouyang, X.~Mu, and R.~Schober, ``Near-field communications: {A} tutorial review,'' \emph{IEEE Open J. Commun. Soc.}, vol.~4, pp. 1999--2049, 2023.

\bibitem{NFC-Tutorial-2}
H.~Lu, Y.~Zeng, C.~You, Y.~Han, J.~Zhang, Z.~Wang, Z.~Dong, S.~Jin, C.~Wang, T.~Jiang, X.~You, and R.~Zhang, ``A tutorial on near‑field {XL‑MIMO} communications toward {6G},'' \emph{IEEE Commun. Surv. Tutor.}, vol.~26, no.~4, pp. 2213--2257, Aug. 2024.

\bibitem{ammar2024handoffs}
H.~A. Ammar, R.~Adve, S.~Shahbazpanahi, G.~Boudreau, and K.~V. Srinivas, ``Handoffs in user‑centric cell‑free {MIMO} networks: A {POMDP} framework,'' \emph{IEEE Trans. Wirel. Commun.}, vol.~23, no.~8, pp. 10\,319--10\,335, Aug. 2024.

\bibitem{ELAA-NFC-1}
X.~L. et~al., ``Multi-user modular {XL-MIMO} communications: Near-field beam focusing pattern and user grouping,'' \emph{IEEE Trans. Wirel. Commun.}, vol.~23, no. 10, Part 2, pp. 13\,766--13\,781, Oct. 2024.

\bibitem{ELAA-NFC-2}
H.~Z. et~al., ``Beam focusing for near-field multiuser {MIMO} communications,'' \emph{IEEE Trans. Wirel. Commun.}, vol.~21, no.~9, pp. 7476--7490, Sep. 2022.

\bibitem{ELAA-NFC-3}
Y.~Zhang and C.~You, ``{SWIPT} in mixed near- and far-field channels: Joint beam scheduling and power allocation,'' \emph{IEEE J. Sel. Areas Commun.}, vol.~42, no.~6, pp. 1683--1698, Jun. 2024.

\bibitem{ELAA-NFC-4}
Z.~T. et~al., ``Dynamic precoding for near‑field secure communications: Implementation and performance analysis,'' \emph{IEEE Internet Things J.}, May 2025.

\bibitem{hassan2020user}
N.~Hassan, M.~T. Hossan, and H.~Tabassum, ``User association in coexisting {RF} and terahertz networks in {6G},'' in \emph{Proc. 2020 IEEE Canadian Conference on Electrical and Computer Engineering (CCECE)}, London, ON, Canada, Aug. 2020, pp. 1--5.

\bibitem{fang2021hybrid}
C.~e.~a. Fang, ``Hybrid precoding in cooperative millimeter wave networks,'' \emph{IEEE Trans. on Wirel. Commun.}, vol.~20, no.~8, pp. 5373--5388, 2021.

\bibitem{MBN-Mobility-2}
Y.~Wang, D.~A. Basnayaka, X.~Wu, and H.~Haas, ``Optimization of load balancing in hybrid {LiFi/RF} networks,'' \emph{IEEE Trans. on Commun.}, vol.~65, no.~4, pp. 1708--1720, 2017.

\bibitem{gupta2024forecaster}
M.~G. et~al., ``Forecaster‑aided user association and load balancing in multi‑band mobile networks,'' \emph{IEEE Trans. Wirel. Commun.}, vol.~23, no.~5, pp. 5157--5171, May 2024.

\bibitem{zeng2020realistic}
Z.~Zeng, M.~D. Soltani, Y.~Wang, X.~Wu, and H.~Haas, ``Realistic indoor hybrid {WiFi} and {OFDMA}‑based {LiFi} networks,'' \emph{IEEE Trans. Commun.}, vol.~68, no.~5, pp. 2978--2991, May 2020.

\bibitem{farokhi2018mobility}
M.~Farokhi, A.~Zolghadrasli, and N.~M. Yamchi, ``Mobility-based cell and resource allocation for heterogeneous ultra-dense cellular networks,'' \emph{IEEE Access}, vol.~6, pp. 66\,940--66\,953, 2018.

\bibitem{HoCost-1}
\BIBentryALTinterwordspacing
M.~A. Dastgheib, H.~Beyranvand, J.~A. Salehi, and M.~Maier, ``Mobility-aware resource allocation in {VLC} networks using {T}-step look-ahead policy,'' \emph{J. Lightwave Technol.}, vol.~36, no.~23, pp. 5358--5370, Dec 2018. [Online]. Available: \url{https://opg.optica.org/jlt/abstract.cfm?URI=jlt-36-23-5358}
\BIBentrySTDinterwordspacing

\bibitem{prado2023enabling}
A.~e.~a. Prado, ``Enabling proportionally-fair mobility management with reinforcement learning in {5G} networks,'' \emph{IEEE J. Sel. Areas Commun.}, vol.~41, no.~6, pp. 1845--1858, 2023.

\bibitem{ren2023handoff}
Q.~R. et~al., ``Handoff-aware distributed computing in high altitude platform station {(HAPS)}--assisted vehicular networks,'' \emph{IEEE Trans. Wirel. Commun.}, vol.~22, no.~12, pp. 8814--8827, Dec. 2023.

\bibitem{bao2017optimizing}
W.~Bao and B.~Liang, ``Optimizing cluster size through handoff analysis in user-centric cooperative wireless networks,'' \emph{IEEE Trans. Wirel. Commun.}, vol.~17, no.~2, pp. 766--778, Feb. 2018.

\bibitem{aboagye2023energy}
S.~Aboagye, T.~M. Ngatched, O.~A. Dobre, and H.~V. Poor, ``Energy-efficient resource allocation for aggregated {RF/VLC} systems,'' \emph{IEEE Trans. on Wirel. Commun.}, vol.~22, no.~10, pp. 6624--6640, 2023.

\bibitem{MBN-TRx-1}
W.~B. et~al., ``Photonics-assisted millimeter-wave multiband integrated sensing and communication system using coherent receiving,'' \emph{IEEE J. Sel. Top. Quantum Electron.}, vol.~29, no.~6, pp. 1--11, Nov. 2023.

\bibitem{MBN-TRx-2}
I.~F. e.~a. Akyildiz, ``A new cubesat design with reconfigurable multi-band radios for dynamic spectrum satellite communication networks,'' \emph{{Ad Hoc Networks}}, vol.~86, pp. 166--178, 2019.

\bibitem{THz-Doppler}
A.~Liao, Z.~Gao, D.~Wang, H.~Wang, H.~Yin, D.~W.~K. Ng, and M.-S. Alouini, ``Terahertz ultra-massive {MIMO}-based aeronautical commun. in space-air-ground integrated networks,'' \emph{IEEE J. Sel. Areas Commun.}, vol.~39, no.~6, pp. 1741--1767, 2021.

\bibitem{Ch-Pred-1}
F.~P. et~al., ``A novel mobility induced channel prediction mechanism for vehicular communications,'' \emph{IEEE Trans. Wirel. Commun.}, vol.~22, no.~5, pp. 3488--3502, May 2022.

\bibitem{Ch-Pred-2-EKF}
W.~M. et~al., ``{UAV}-assisted communications in {SAGIN-ISAC}: Mobile user tracking and robust beamforming,'' \emph{IEEE J. Sel. Areas Commun.}, vol.~43, no.~1, pp. 186--200, Jan. 2025.

\bibitem{Ch-Pred-3-EKF-UserTracking}
J.~F. et~al., ``Mobility-aware predictive beamforming design in sensing-assisted {SWIPT} systems,'' \emph{IEEE Wireless Commun. Lett.}, vol.~14, no.~5, pp. 1476--1480, May 2025.

\bibitem{Ch-Pred-4}
D.~Y. et~al., ``Temporal-frequency domain channel prediction for {LEO} satellite communication system: {A} novel {TFformer} structure,'' \emph{IEEE Trans. Wirel. Commun.}, 2025, early Access.

\bibitem{barahimi2024rscnet}
B.~B. et~al., ``{RSCNet}: Dynamic {CSI} compression for cloud‑based {WiFi} sensing,'' pp. 4179--4184, Jun. 2024.

\bibitem{THz-Channel-2-HITRAN}
``The {HITRAN2020} molecular spectroscopic database,'' \emph{J. Quant. Spectrosc. Radiat. Transf.}, vol. 277, p. 107949, 2022.

\bibitem{yan2022energy}
L.~e.~a. Yan, ``Energy-efficient dynamic-subarray with fixed true-time-delay design for terahertz wideband hybrid beamforming,'' \emph{IEEE J. Sel. Areas Commun.}, vol.~40, no.~10, pp. 2840--2854, 2022.

\bibitem{hossan2021mobility}
M.~T. Hossan and H.~Tabassum, ``Mobility-aware performance in hybrid {RF} and terahertz wireless networks,'' \emph{IEEE Trans. on Commun.}, vol.~70, no.~2, pp. 1376--1390, 2021.

\bibitem{saeidi2023tractable}
M.~A. Saeidi, H.~Shoaib, and H.~Tabassum, ``A tractable handoff-aware rate outage approximation with applications to {THz}-enabled vehicular network optimization,'' in \emph{Proc. 2023 IEEE Global Communications Conference (GLOBECOM)}.\hskip 1em plus 0.5em minus 0.4em\relax Kuala Lumpur, Malaysia: IEEE, Dec. 2023, pp. 5092--5097.

\bibitem{Blockage-aware1}
D.~Kumar, J.~Kaleva, and A.~Tölli, ``Blockage-aware reliable {mmWave} access via coordinated multi-point connectivity,'' \emph{IEEE Trans. Wirel. Commun.}, vol.~20, no.~7, pp. 4238--4252, Jul. 2021.

\bibitem{AntennaGain}
S.~A. Busari, K.~M.~S. Huq, S.~Mumtaz, J.~Rodriguez, Y.~Fang, D.~C. Sicker, S.~Al-Rubaye, and A.~Tsourdos, ``Generalized hybrid beamforming for vehicular connectivity using {THz} massive {MIMO},'' \emph{IEEE Trans. Veh. Technol.}, vol.~68, no.~9, pp. 8372--8383, 2019.

\bibitem{ye2021modeling}
J.~e.~a. Ye, ``Modeling co-channel interference in the {THz} band,'' \emph{IEEE Trans. Veh. Technol.}, vol.~70, no.~7, pp. 6319--6334, 2021.

\bibitem{pradhan2023robust}
A.~Pradhan, M.~A. Abd‑Elmagid, H.~S. Dhillon, and A.~F. Molisch, ``Robust optimization of {RIS} in terahertz under extreme molecular re‑radiation manifestations,'' \emph{IEEE Trans. Wirel. Commun.}, vol.~23, no.~3, pp. 1783--1797, Mar. 2024.

\bibitem{HBF-1}
F.~Sohrabi and W.~Yu, ``Hybrid analog and digital beamforming for {mmWave} {OFDM} large-scale antenna arrays,'' \emph{IEEE J. Sel. Areas Commun.}, vol.~35, no.~7, pp. 1432--1443, 2017.

\bibitem{big-M}
M.~Cococcioni and L.~Fiaschi, ``The big-{M} method with the numerical infinite {M},'' \emph{Optimization Letters}, vol.~15, no.~7, pp. 2455--2468, 2021.

\bibitem{FP-Part1}
K.~Shen and W.~Yu, ``Fractional programming for communication systems—{Part I}: Power control and beamforming,'' \emph{IEEE Trans. Signal Process.}, vol.~66, no.~10, pp. 2616--2630, May 2018.

\bibitem{FP-Part2}
------, ``Fractional programming for communication systems—{Part II}: Uplink scheduling via matching,'' \emph{IEEE Trans. Signal Process.}, vol.~66, no.~10, pp. 2631--2644, May 2018.

\bibitem{cvx}
M.~Grant and S.~Boyd, ``{{CVX}: Matlab Software for disciplined convex programming, version 2.1},'' Available at \url{http://cvxr.com/cvx}, Mar. 2014.

\bibitem{mosek}
\BIBentryALTinterwordspacing
M.~ApS, ``The optimizer for mixed-integer problems,'' 2023. [Online]. Available: \url{https://docs.mosek.com/latest/cxxfusion/mip-optimizer.html}
\BIBentrySTDinterwordspacing

\bibitem{MM-Paper}
Y.~S. et~al., ``Majorization-minimization algorithms in signal processing, communications, and machine learning,'' \emph{IEEE Trans. Signal Process.}, vol.~65, no.~3, pp. 794--816, Feb. 2017.

\bibitem{Hyb-RSMA-NOMA}
M.~A. Saeidi and H.~Tabassum, ``Resource allocation and performance analysis of hybrid {RSMA-NOMA} in the downlink,'' in \emph{2023 IEEE 34th Annual International Symposium on Personal, Indoor and Mobile Radio Commun. (PIMRC)}, 2023, pp. 1--6.

\bibitem{razaviyayn2013unified}
M.~e.~a. Razaviyayn, ``A unified convergence analysis of block successive minimization methods for nonsmooth optimization,'' \emph{SIAM Journal on Optimization}, vol.~23, no.~2, pp. 1126--1153, 2013.

\bibitem{ye2011interior}
Y.~Ye, \emph{Interior point algorithms: theory and analysis}.\hskip 1em plus 0.5em minus 0.4em\relax John Wiley \& Sons, 2011.

\bibitem{andersen2003implementing}
E.~D. e.~a. Andersen, ``On implementing a primal-dual interior-point method for conic quadratic optimization,'' \emph{Mathematical Programming}, vol.~95, pp. 249--277, 2003.

\bibitem{sun2020movement}
W.~S. et~al., ``Movement aware {CoMP} handover in heterogeneous ultra-dense networks,'' \emph{IEEE Trans. Commun.}, vol.~69, no.~1, pp. 340--352, Jan. 2021.

\bibitem{boyd2004convex}
S.~Boyd, S.~P. Boyd, and L.~Vandenberghe, \emph{Convex optimization}.\hskip 1em plus 0.5em minus 0.4em\relax Cambridge University Press, 2004.

\bibitem{marler2010weighted}
R.~T. Marler and J.~S. Arora, ``The weighted sum method for multi-objective optimization: new insights,'' \emph{Structural and multidisciplinary optimization}, vol.~41, pp. 853--862, 2010.

\bibitem{interdonato2020local}
G.~Interdonato, M.~Karlsson, E.~Bj{\"o}rnson, and E.~G. Larsson, ``Local partial zero-forcing precoding for cell-free massive {MIMO},'' \emph{IEEE Trans. Wirel. Commun.}, vol.~19, no.~7, pp. 4758--4774, 2020.

\bibitem{deshpande2022wideband}
N.~Deshpande, M.~R. Castellanos, S.~R. Khosravirad, J.~Du, H.~Viswanathan, and R.~W. Heath, ``A wideband generalization of the near-field region for extremely large phased-arrays,'' \emph{IEEE Wireless Commun. Lett.}, vol.~12, no.~3, pp. 515--519, Mar. 2023.

\bibitem{ImpCSI-1}
X.~M. et~al., ``Cooperative beamforming for {RIS}-aided cell-free massive {MIMO} networks,'' \emph{IEEE Trans. Wirel. Commun.}, vol.~22, no.~11, pp. 7243--7258, Nov. 2023.

\bibitem{ImpCSI-2}
G.~Z. et~al, ``A framework of robust transmission design for {IRS}-aided {MISO} communications with imperfect cascaded channels,'' \emph{IEEE Trans. Signal Process.}, vol.~68, no.~17, pp. 5092--5106, Sep. 2020.

\bibitem{alizadeh2023power}
M.~Alizadeh and H.~Tabassum, ``Power control with {QoS} guarantees: A differentiable projection-based unsupervised learning framework,'' \emph{IEEE Trans. Commun.}, vol.~71, no.~8, pp. 4605--4619, Aug. 2023.

\end{thebibliography}
\bibliographystyle{IEEEtran}
\begin{IEEEbiography}[{\includegraphics[width=1in,height=1.25in, clip,keepaspectratio]{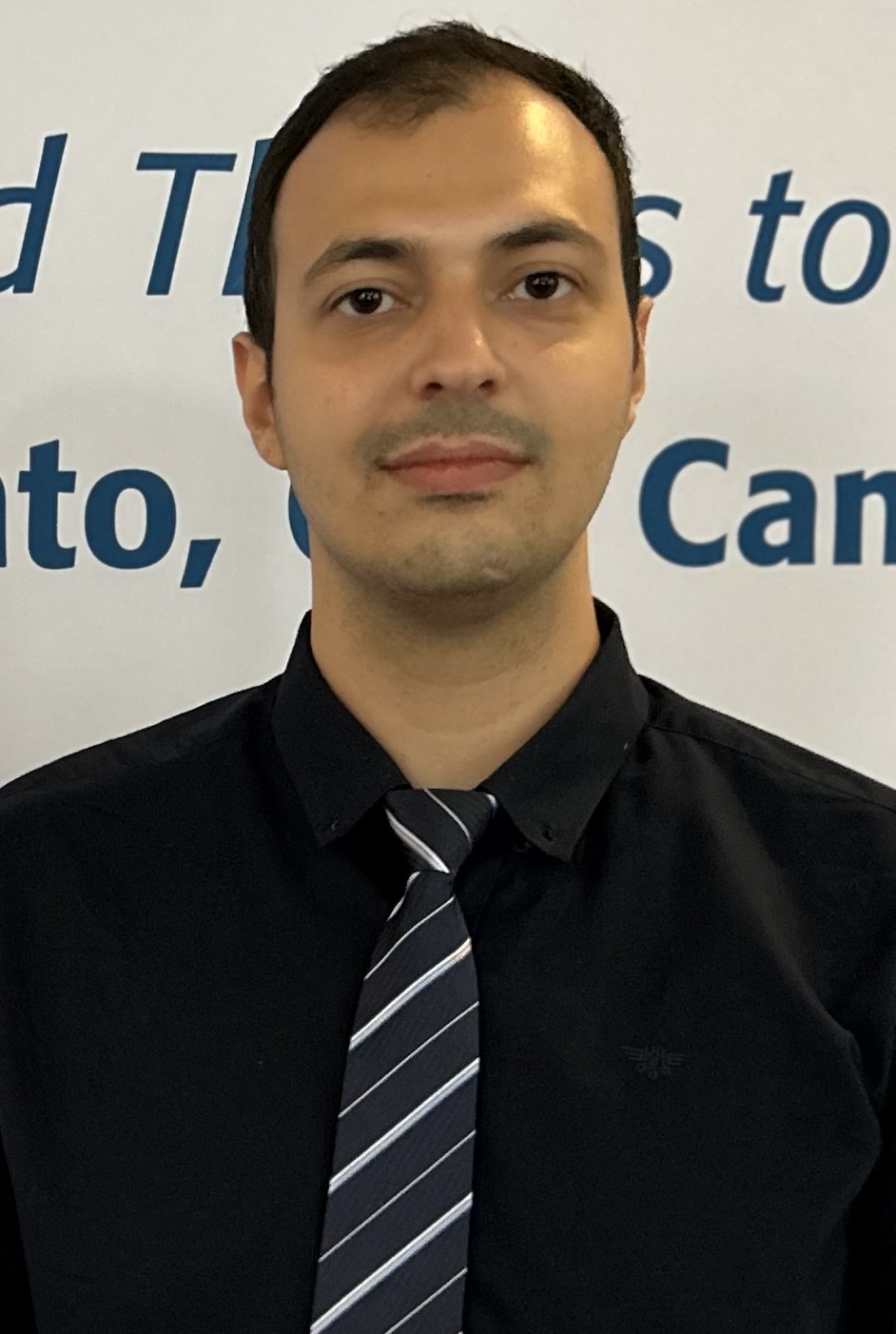}}]{\textbf{Mohammad Amin Saeidi}} (Graduate Student Member, IEEE) received the M.Sc. degree in Electrical Engineering – Communication Systems from Amirkabir University of Technology, Tehran, Iran, in 2021. He is currently pursuing a Ph.D. degree in Electrical Engineering and Computer Science at York University, Canada. His research focuses on resource and mobility management, as well as optimization in 6G wireless communications, terahertz communication, multi-band networks, and reconfigurable intelligent surfaces. He has served as a reviewer for various IEEE journals, including IEEE Transactions on Wireless Communications, IEEE Transactions on Communications, IEEE Transactions on Mobile Computing, IEEE Open Journal of the Communications Society, IEEE Wireless Communications Letters, IEEE Communications Letters, and IEEE Transactions on Green Communications and Networking.
\end{IEEEbiography}

\begin{IEEEbiography}[{\includegraphics[width=1in,height=1in,clip,keepaspectratio]{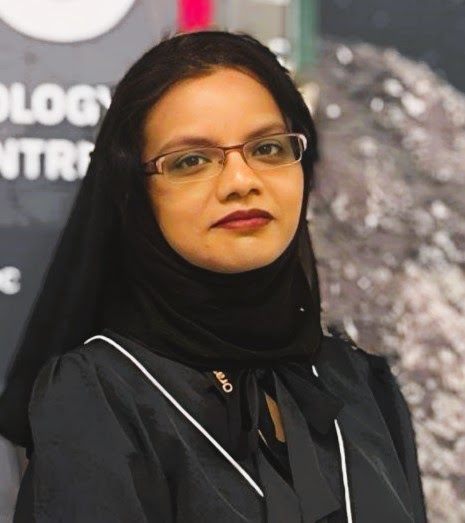}}]{Hina Tabassum} \,
(Senior Member, IEEE) \,  \, (M'12-SM'18)
 received the Ph.D. degree from the King Abdullah University of Science and Technology, Thuwal, Saudi Arabia. She is currently an Associate Professor with the Lassonde School of Engineering, York University, Toronto, where she joined as an Assistant Professor in 2018. She was appointed as a Visiting Faculty with University of Toronto, Toronto, ON, Canada, in 2024, and the York Research Chair of 5G/6G-enabled mobility and sensing applications in 2023, for five years. She has coauthored more than 120 refereed articles in well-reputed IEEE journals, magazines, and conferences. Her current research interests include multi-band 6G wireless communications and sensing networks, connected and autonomous systems, AI-enabled network mobility, and resource management solutions. She has been selected as the IEEE ComSoc Distinguished Lecturer for the term 2025–2026. She is listed in Stanford’s list of the World’s Top Two-Percent Researchers in 2021–2024. She was the recipient of the Lassonde Innovation Early-Career Researcher Award in 2023 and the N2Women: Rising Stars in Computer Networking and Communications in 2022. She was the Founding Chair of the Special Interest Group on THz communications in the IEEE Communications Society–Radio Communications Committee. Currently, she is serving as an Area Editor of the IEEE Open Journal of the Communications Society and IEEE Communications Surveys and Tutorials as well as an Associate Editor for IEEE Transactions on Communications, IEEE Transactions on Mobile Computing, and IEEE Transactions on Wireless Communications.
\end{IEEEbiography}
\end{document}